\documentclass[11pt, a4paper, oneside, reqno]{amsart}

\usepackage[latin1 ]{inputenc}
\usepackage{amsmath, amsthm, amssymb, amsfonts}
\usepackage[foot]{amsaddr}
\usepackage[english]{babel}
\usepackage{bbding}
\usepackage{booktabs}
\usepackage{cite}
\usepackage[colorlinks=true, linkcolor = blue, citecolor = dgreen]{hyperref}
\usepackage{cleveref}
\usepackage{color}
\usepackage{crossreftools}
\usepackage{diagbox}
\usepackage{dsfont}
\usepackage[shortlabels]{enumitem}
\usepackage{eucal}
\usepackage{float}
\usepackage[a4paper]{geometry} 
\usepackage{graphicx}
\usepackage{mathrsfs}
\usepackage{multirow}
\usepackage[all]{xy} 
\usepackage[numbers]{natbib}
\usepackage{rotating}
\usepackage{setspace}
\usepackage{stmaryrd}
\usepackage{subcaption}
\usepackage[normalem]{ulem}
\usepackage{xspace}
\usepackage{url}
\usepackage[dvipsnames]{xcolor}

\setlist{leftmargin=*, topsep=0.5em, parsep=0pt, itemsep=1em, labelindent=0pt, align=left}

\theoremstyle{definition}
\newtheorem{theorem}{Theorem}[section]

\newtheorem{definition}{Definition}[section]

\newtheorem{proposition}{Proposition}[section]

\newtheorem{remark}{Remark}[section]

\newtheorem{assumption}{Assumption}%[section]
\crefname{assumption}{assumption}{assumptions}

\newtheorem{assumptionm}{Assumption}%[section]
\crefname{assumption}{assumption}{assumptions}

\crefname{assumptionstar}{assumption}{assumptions}

\newtheorem{lemma}{Lemma}[section]

\numberwithin{equation}{section}

\newtheoremstyle{blue}%
{6.5pt}% Space above
{5pt}% Space below 
{}% Body font
{}% Indent amount
{\bfseries\color{blue}}% Theorem head font
{.}% Punctuation after theorem head
{.4em}% Space after theorem head 
{}% Theorem head spec (can be left empty, meaning ‘normal’)
\theoremstyle{blue}

\newcommand{\bigO}[1]{\mathcal{O}(#1)\xspace}
\newcommand{\littleo}[1]{o(#1)\xspace}

\newcommand{\bvec}[1]{\overset{{}_{\shortleftarrow}}{#1}}

\newcommand{\Qro}{\mathbb{Q}}
\newcommand{\Ex}{\mathbb{E}}

\newcommand{\esssup}{\mathrm{ess\;sup}}

\renewcommand{\AA}{\mathcal{A}}
\newcommand{\BB}{\mathcal{B}}
\newcommand{\CC}{\mathcal{C}}

\newcommand{\FF}{\mathcal{F}}
\newcommand{\GG}{\mathcal{G}}

\newcommand{\LL}{\mathcal{L}}

\newcommand{\NN}{\mathcal{N}}
\newcommand{\PP}{\mathcal{P}}
\renewcommand{\SS}{\mathcal{S}}
\newcommand{\WW}{\mathcal{W}}
\newcommand{\XX}{\mathcal{X}}
\newcommand{\YY}{\mathcal{Y}}
\newcommand{\ZZ}{\mathcal{Z}}

\newcommand{\AAscr}{\mathscr{A}}
\newcommand{\BBscr}{\mathscr{B}}
\newcommand{\CCscr}{\mathscr{C}}
\newcommand{\LLscr}{\mathscr{L}}

\newcommand{\dd}{\mathrm{d}}
\newcommand{\bd}{\bvec{\dd}}

\newcommand{\mrB}{\mathring{B}}

\newcommand{\reals}{\mathbb{R}}

\newcommand{\naturals}{\mathbb{N}}
\newcommand{\integers}{\mathbb{Z}}

\newcommand{\vep}{\varepsilon}

\newcommand{\tinT}{t\in[0,T]}
\newcommand{\intr}{\mathfrak{r}}

\newcommand{\gind}[2]{\textbf{1}_{#1}(#2)}
\newcommand{\rind}[1]{\textbf{1}_{#1}}

\numberwithin{table}{section}

\definecolor{red}{HTML}{D62728}
\definecolor{blue}{RGB}{ 0, 109, 219}
\definecolor{dgreen}{rgb}{0,.8,0}

\newcommand{\recbullet}{\textcolor{black}{\raisebox{.45ex}{\rule{.8ex}{.8ex}}}}

\crefname{enumi}{}{}
\Crefname{enumi}{}{}

\crefformat{enumi}{(#2#1#3)}
\Crefformat{enumi}{(#2#1#3)}

\crefrangeformat{enumi}{(#3#1#4) to~(#5#2#6)}
\Crefrangeformat{enumi}{(#3#1#4) to~(#5#2#6)}

\crefmultiformat{enumi}{(#2#1#3)}{ and~(#2#1#3)}{, (#2#1#3)}{ and ~(#2#1#3)}
\Crefmultiformat{enumi}{(#2#1#3)}{ and~(#2#1#3)}{, (#2#1#3)}{ and ~(#2#1#3)}

\crefrangeformat{equation}{eqs. (#3#1#4)--(#5#2#6)}

\begin{document}
\title[]{On Stochastic Partial Differential Equations and their applications to Derivative Pricing through a conditional Feynman-Kac formula}
\author{Kaustav Das$^{\dagger \ddagger}$}
\author{Ivan Guo$^{\dagger \ddagger}$}
\author{Gr\'egoire Loeper$^{\mathsection}$}
\address{$^\dagger$School of Mathematics, Monash University, Victoria, 3800 Australia.}
\address{$^\ddagger$Centre for Quantitative Finance and Investment Strategies, Monash University, Victoria, 3800 Australia.}
\address{$^\mathsection$BNP Paribas Global Markets, Paris, France.}
\email{kaustav.das@monash.edu, ivan.guo@monash.edu, gregoire.loeper@bnpparibas.com}
\date{}
\maketitle

\begin{abstract}

The price of a financial derivative can be expressed as an iterated conditional expectation, where the inner term conditions on the future of an auxiliary process. We show that this inner conditional expectation solves an SPDE (a `conditional Feynman-Kac formula'). The problem requires conditioning on a backward filtration generated by the noise of the auxiliary process and enlarged by its terminal value, leading us to search for a backward Brownian motion here. This adds a source of irregularity to the SPDE which we tackle with new techniques. Lastly, we establish a new class of mixed Monte-Carlo PDE numerical methods.

\vspace{.5cm}

\hspace{-.42cm}Keywords: Stochastic PDE; Conditional Feynman-Kac Formula; Mixed Monte-Carlo PDE; Backward Stochastic Calculus; Stochastic Volatility.

\end{abstract}

%----------------------------------------------------------------------------Introduction-----------------------------------------------------------------------------------
%-----------------------------------------------------------------------------------------------------------------------------------------------------------------------------
%-----------------------------------------------------------------------------------------------------------------------------------------------------------------------------
\section{Introduction}
\label{sec:introduction}
\noindent 
The purpose of this article is to demonstrate that certain types of Stochastic Partial Differential Equations (SPDEs) naturally arise in financial derivative pricing. Briefly, let $X$, $V$, and $\intr$ be the asset price process, an auxiliary process (often stochastic variance/volatility), and deterministic interest rate respectively, see \Cref{sec:preliminaries} for their definitions, and \Cref{sec:multivariablesetting} for the general multivariable setting. Let $H$ be a European derivative which pays $\varphi(X_T)$ at time $T$. One can express $H_t$ as an iterated conditional expectation under a chosen risk-neutral measure in the following fashion: 
\begin{align*}
		H_t &= e^{-\int_t^T \intr_s \dd s}\,\Ex \big [u(t, X_t) | X_t, V_t \big ],
\end{align*}
where
\begin{align}
	u(t, x) := \Ex[ \varphi(X_T) | X_t = x, \GG_{t,T}]. \label{eqn:utx}
\end{align}
Here we write $\GG_{t,T}$ as a placeholder which will be given precise meaning later on, but roughly speaking it is a suitable $\sigma$-algebra which essentially corresponds to the future of the auxiliary process $V$ over $[t, T]$. Thus $u(t, x)$ is a random field which is $\GG_{t, T}$ measurable for each fixed $(t, x)$. Denoting by $V_{[t, T]}$ the trajectory of $V$ over $[t, T]$, then at least informally, one can think of $u(t, x)$ as a functional of $V_{[t, T]}$, namely $u(t, x) \equiv u(t, x, V_{[t, T]})$.

In this article we prove that $u(t,x)$ from \cref{eqn:utx} solves a backward linear SPDE, similar to the classical Feynman-Kac formula from the deterministic PDE scenario. Such a relationship is known as a conditional Feynman-Kac formula, and many versions of these formulas have been studied in the literature, albeit in the context of non-linear filtering theory. Recently, these results have been exploited in the context of generative modelling, see \citep{ho2020denoising}. Naturally, the existence and regularity properties of these types of SPDEs that arise through conditional Feynman-Kac formulas are of great importance. We remark that the backward SPDEs considered in this article are understood in the backward It\^o sense, and thus are not related to the theory of backward stochastic differential equations (BSDEs) established by Pardoux and Peng \citep{pardoux1990adapted}, which has become quite prevalent in the current stochastic analysis literature. On this note, a number of recent articles such as \citep{bayer2022pricing,bank2023rough} utilise backward SPDEs to represent prices of financial derivatives. However, the filtration they condition on is forward and the stochastic integration is forward {\textemdash} consequently the types of backward SPDEs they study are of the Pardoux and Peng type. Thus their methodology is completely different to ours.

In the non-linear filtering literature there is a shift in terminology. Namely, one considers a signal process $X$ that is unobserved, and an observation process $V$ which is observed, these processes being obtained through an SDE. The objective is to find an SPDE representation for conditional expectations of the form \cref{eqn:utx}, i.e., a conditional Feynman-Kac formula. However, these formulas depend on the precise formulation of the SDE for $(X, V)$ as well as the explicit definition of the $\sigma$-algebra $\GG_{t, T}$. Additionally, the succinct martingale arguments typically used in modern proofs for classical Feynman-Kac formulas from the deterministic PDE setting cannot be utilised to derive conditional Feynman-Kac formulas, as the collection of $\sigma$-algebras $(\sigma(X_t) \vee \GG_{t,T})_{\tinT}$ are neither increasing nor decreasing, meaning that $t \mapsto u(t, X_t)$ does not form a Doob martingale. Thus, more sophisticated methods must be employed.

We now briefly outline the various results on conditional Feynman-Kac formulas that have arisen in the non-linear filtering literature. \citet[][Theorem 2.1]{pardoux1979stochastic} deduces a conditional Feynman-Kac formula for a system where the noises driving the signal process $X$ and observation process $V$ are correlated, yet the coefficients in the SDE for $X$ do not depend on $V_t$, and where $\GG_{t, T}$ corresponds to the increments of $V$ over $[t, T]$. Due to the backward and forward nature of the problem, typical stochastic analysis theory cannot be utilised, and they resort to a direct time discretisation method in their proof. \citet[][Theorem 4.1]{krylov1982stochastic} deduce a conditional Feynman-Kac formula similar to the one from \citep{pardoux1979stochastic}, albeit with a more elegant proof involving a clever application of the classical Feynman-Kac formula in tandem with orthogonality arguments. \citet[][Theorem 6.5]{pardoux1982equations} extends the aforementioned results, namely a conditional Feynman-Kac formula is established in the case where the coefficients in the SDE for the signal process $X$ can depend on the observation process $V_t$, and moreover $\GG_{t, T}$ now refers to the path of $V$ over $[t, T]$. Unfortunately the elegant methods from \citep{krylov1982stochastic} cannot be applied here; roughly speaking this is because the dependence of the coefficients on $V_t$ precludes their particular use of the classical Feynman-Kac formula alongside orthogonality arguments. Thus the time discretisation method from \citep{pardoux1979stochastic} must be appealed to and modified accordingly. \citet[][Theorem 4.2]{ocone1993stochastic} consider the case for when the coefficients in the SDE of the signal process $X$ depend on the whole trajectory of the observation process $V$, and moreover, the $\sigma$-algebra $\GG_{t, T} \equiv \GG_{0, T}$ refers to the path of $V$ over $[0, T]$. Due to this framework, the anticipating stochastic calculus must be utilised, and thus the conditional Feynman-Kac formula they derive involves Skorokhod integrals. \citet{pardoux1994backward} show that so-called Backward Doubly Stochastic Differential Equations (Backward in the sense of Pardoux and Peng) can be utilised to represent solutions to certain Backward (in the sense of It\^o) semilinear SPDEs. In this situation a conditional Feynman-Kac formula comes as a particular case of their methodology. Furthermore, their methodology generalises previous conditional Feynman-Kac formulas as an additional term ($a$ in \citep[][Remark 3.4]{pardoux1994backward}) allows for some added flexibility. However, their methodology does not allow for correlated Brownian motions, and thus in another way is more restrictive than previously developed conditional Feynman-Kac formulas. 

In this article we prove a version of the conditional Feynman-Kac formula corresponding to the derivative pricing problem outlined in the first paragraph, and moreover we study and prove results on the existence and regularity of the associated SPDE. Our problem differs to the ones considered in the non-linear filtering literature as the $\sigma$-algebra $\GG_{t, T}$ we are required to condition on involves the increments of the noise driving the auxiliary process $V$ over $[t, T]$, as well as a value of $V$ on this interval. Additionally, our coefficients in the SDE for the asset price process $X$ can depend on the auxiliary process $V_t$. The implications of this are that when deriving the associated SPDE for our conditional Feynman-Kac formula, one must search for a new backward Brownian motion in this particular backward filtration $(\GG_{t, T})_{\tinT}$. This in turn adds an additional source of irregularity to the SPDE which we tackle with new techniques. We choose to consider such a setting as the financial applications demand this.

An important application of our conditional Feynman-Kac formula is in the development of a mixed Monte-Carlo PDE method for pricing financial derivatives. Indeed, we see from \cref{eqn:utx} that the time $0$ price of a European derivative is given by
\begin{align*}
	H_0 = e^{-\int_0^T \intr_r \dd r} \Ex [ u(0, x) ].
\end{align*}
Through our conditional Feynman-Kac formula, $u(t, x)$ solves a SPDE. Thus the basic idea for a mixed Monte-Carlo PDE method is to simulate the price $H_0$ by numerically solving the SPDE repeatedly to obtain many i.i.d. copies of $u(0, x)$, and then simply averaging over them. To contrast this approach with other well known methods, we first note that closed-form formulas for $H_0$ are rare. Thus the standard practice in applications is to calculate $H_0$ through numerical methods, either via a Full Monte-Carlo simulation, or numerically solving the associated deterministic PDE obtained through the classical Feynman-Kac formula. However, both these methods come with their disadvantages, especially in a high dimensional setting. Namely, Full Monte-Carlo methods suffer from high variance and large computational costs, whereas numerical PDE methods do not fare well for dimensions greater than 3 or 4. Thus the main advantage of a mixed Monte-Carlo PDE method over a Full Monte-Carlo simulation or numerical PDE methods is that one can enjoy the best of both worlds by choosing the system in such a way that one extracts the benefits of each latter method, and discards their disadvantages. For example, suppose our system has $M$ components. A clever use of a mixed Monte-Carlo PDE method could be to pass on one or two components whose paths are known to be quite volatile onto the PDE solver, and the rest $M-1$ or $M-2$ components onto the Monte-Carlo simulation. As PDE methods fare well in lower dimensions, this is efficient, and moreover we achieve variance reduction as compared to a Full Monte-Carlo method as the volatile paths have been tackled by the PDE solver. In short, mixed Monte-Carlo PDE methods serve to provide variance and dimensionality reduction for the derivative pricing problem. 

Mixed Monte-Carlo PDE methods have recently seen a surge of interest in the literature, and this is mainly due to their benefits in applications being immense. These methods were initiated by \citet{loeper2009mixed} and \citet{lipp2014mixing}, who develop a mixed Monte-Carlo PDE method by showing that one can express the price of a derivative as an expectation of a function that solves a PDE with random coefficients. They coin the term `conditional PDE' to refer to these types of PDEs. This idea is then built upon by \citet{dang2015dimension, dang2017dimension}, who combine Fourier transform methods in order to obtain quasi closed-formed formulas for the solutions to these conditional PDEs in the context of various pricing problems. \citet{cozma2016mixed} prove a number of theoretical results regarding error and computational runtime for these mixed Monte-Carlo PDE methods. Most recently, \citet{farahany2020mixing} consider a mixed Monte-Carlo PDE method for the pricing of Bermudan options, effectively writing the continuation value as an iterated conditional expectation, then deducing that the inner one solves a conditional PDE. However, what all these aforementioned mixed Monte-Carlo PDE methods share in common is that the `PDE' aspect refers to a conditional PDE. Instead, in this article we develop a mixed Monte-Carlo PDE method where the `PDE' aspect now refers to an SPDE, which we believe is the first of its kind. Our mixed Monte-Carlo PDE method thus serves as a link between the field of derivative pricing and SPDEs; we hope that this connection will yield further insights in future research.

Our first main result is \Cref{thm:spdeexistence}, which pertains to the existence and regularity of the SPDE of interest. Our next main result is \Cref{thm:conditionalfeynmanwellposed}, which is a conditional Feynman-Kac formula. Lastly, we showcase the utility of the conditional Feynman-Kac formula by providing a simple demonstration of a mixed Monte-Carlo PDE method for pricing a European option in \Cref{sec:numerics}. The sections are organised as follows:
\begin{itemize}
\item \Cref{sec:preliminaries} contains preliminary content, where we provide the model framework and introduce the SPDE which shall be the focus of this article.
\item In \Cref{sec:mainresults} we provide our main results, namely the existence of a unique solution to the aforementioned SPDE, as well as a conditional Feynman-Kac formula.
\item \Cref{sec:proofs} is devoted to the proofs of our main results from \Cref{sec:mainresults}.
\item \Cref{sec:multivariablesetting} consists of extensions of our main results to the multivariable setting.
\item In \Cref{sec:numerics} we explore a numerical example for pricing a European option by mixing numerical PDE and Monte-Carlo methods via our conditional Feynman-Kac formula.
\end{itemize}
\Cref{appen:backwarddefns} contains some content regarding backward stochastic calculus which we will extensively utilise. We remark that the backward stochastic calculus theory we consider when studying backward SPDEs in this article should not be confused with the theory of BSDEs initiated by Pardoux and Peng.
%--------------------------------------------Informal derivation of conditional Feynman-Kac formula-------------------------------------------------------
%-----------------------------------------------------------------------------------------------------------------------------------------------------------------------------

\subsection{Informal derivation of conditional Feynman-Kac formula}

As motivation for the rest of the article, we will now provide an informal argument which elucidates how the SPDE in the conditional Feynman-Kac formula arises, and which moreover, highlights some of the main ideas in the (rather technical) proof of it (\Cref{prop:conditionalfeynmanwellposedprop} and \Cref{thm:conditionalfeynmanwellposed}). Definitions of terminology, objects and notation in the following can be found in \Cref{sec:preliminaries}. Leading on from the first paragraph, recall $H$ is the price of a European derivative which pays $\varphi(X_T)$ at time $T$. Consider the following backward SPDE
\begin{align}
\begin{split}
	-\dd u (t, x) &= \left (\LL^x_t - \CC^x_t \right ) u(t,x)\dd t + \BB^x_t u(t,x) \bd B_t, \\
	 u(T,x) &= \varphi(x), \label{eqn:spdeintro}
\end{split}
\end{align}
where we have the following family of (stochastic) differential operators indexed by $t \in [0, T]$,
\begin{align}
\begin{split}
	\LL^x_t &:= \frac{1}{2} \sigma^2(t,x,V_t) \partial_x^2 + \mu (t, x, V_t) \partial_x, \\
	\BB^x_t &:=  \rho_t \sigma(t,x,V_t) \partial_x, \\
	\CC^x_t &:= \rho_t \beta(t,V_t) \sigma_y(t, x, V_t) \partial_x.
	\label{eqn:operatorsintro}
\end{split}
\end{align}
The coefficients in the operators \cref{eqn:operatorsintro} stem from the system \crefrange{eqn:system2X}{eqn:system2V}. Moreover, the term $\bd B_t$ indicates backward stochastic integration which is defined in \Cref{defn:backstoch}. The goal is to show that the following object
\begin{align*}
	u(t, x) = \Ex[ \varphi(X_T) | X_t = x, \GG_{t,T}],
\end{align*}
solves the SPDE \cref{eqn:spdeintro}, where $\GG_{t, T}$ is a $\sigma$-algebra roughly corresponding to the future of the process $V$. Suppose $u(t, x)$ is the unique solution to the SPDE \cref{eqn:spdeintro}, backward adapted to $(\GG_{t, T})_{\tinT}$. The first thing to note is that it does not make sense to consider the stochastic differential of the mapping $t \mapsto u(t, X_t)$. The reason being is that $X$ corresponds to the solution of a forward SDE, however $(\GG_{t,T})_{\tinT}$ is a backward filtration. Hence if a stochastic differential existed, it would require movements both forward and backward in time, which is not possible within the theory of It\^o. However, it is perfectly legitimate to consider an increment of $t \mapsto u(t, X_t)$ over a finite partition $\{t = t_0 < t_1 < \cdots < t_{n-1} < t_n = T\}$ of $[t, T]$. Write $\Ex_{t, x}^{t, T} \equiv \Ex[\cdot | X_t = x, \GG_{t, T}]$. Furthermore, we note that 
\begin{align}
	\Ex^{t, T}_{t,x} \left [ \sum_{i = 0}^{n-1} u(t_{i+1}, X_{t_{i+1}}) - u(t_i, X_{t_i}) \right ] = \Ex \left [ \varphi(X_T) | X_t = x, \GG_{t, T} \right ] - u(t, x). \label{eqn:telescoping1}
\end{align}
Hence once we show that the LHS of the preceding expression tends to $0$ in $L^1(\Qro_{t, x})$ as $n \to \infty$, then we are done, since the RHS does not depend on $n$. Ergo, it is imperative that we study the increment of $t \mapsto u(t, X_t)$. We do so by utilising the following decomposition:
\begin{align*}
	u(t_{i+1}, X_{t_{i+1}}) - u(t_i, X_{t_i}) &= \left [ u(t_{i+1}, X_{t_{i+1}}) - u(t_{i+1}, X_{t_i}) \right ] + \left [u(t_{i+1}, X_{t_i}) - u(t_i, X_{t_i}) \right ] \\
&= \chi_i + \tau_i,
\end{align*}
where 
\begin{equation*}
\begin{aligned}
	\chi_i &:= u(t_{i+1}, X_{t_{i+1}}) - u(t_{i+1}, X_{t_i}), \qquad & \tau_i &:=u(t_{i+1}, X_{t_i}) - u(t_i, X_{t_i}).
\end{aligned}
\end{equation*}
Notice that for $\chi_i$ space is moving and time is fixed, whereas for $\tau_i$ space is fixed and time is moving. 
%----------chi_i
We can rewrite $\chi_i$ using It\^o's formula: 
\begin{align*}
	\chi_i &=  u(t_{i+1}, X_{t_{i+1}}) - u(t_{i+1}, X_{t_i}) = \int_{t_i}^{t_{i+1}} u_x(t_{i+1}, X_r) \dd X_r + \frac{1}{2} \int_{t_i}^{t_{i+1}} u_{xx}(t_{i+1}, X_r) \dd \langle X, X \rangle_r \\
		&= \int_{t_i}^{t_{i+1}} \left (u_x(t_{i+1}, X_r) \mu(r, X_r, V_r) + \frac{1}{2} u_{xx}(t_{i+1}, X_r) \sigma^2(r, X_r, V_r) \right ) \dd r \\
		&\quad+ \int_{t_i}^{t_{i+1}} u_x(t_{i+1}, X_r) \rho_r \sigma(r, X_r, V_r) \dd B_r + \int_{t_i}^{t_{i+1}} u_x(t_{i+1}, X_r) \varrho_r \sigma(r, X_r, V_r) \dd \hat B_r.
\end{align*}
At this point we note the following two facts. First, the $\dd \hat B$ integral in the preceding expression will not contribute after taking $\Ex_{t,x}^{t, T}$ due to independence of $B$ and $\hat B$. Second, we require the $\dd B$ stochastic integral to be a backward one, due to the measurability properties of the solution $u(t, x)$. Hence we now consider `reversing' the $\dd B$ integral as follows:
\begin{align}
	\int_{t_i}^{t_{i+1}} u_x(t_{i+1}, X_r) \rho_r \sigma(r, X_r, V_r) \dd B_r &= \int_{t_i}^{t_{i+1}} u_x(t_{i+1}, X_r) \rho_r \sigma(r, X_r, V_r) \bd B_r  \nonumber \\
	&\quad- \int_{t_i}^{t_{i+1}} \dd \langle u_x(t_{i+1}, X_\cdot) \rho_\cdot \sigma(\cdot , X_\cdot, V_\cdot) , B_\cdot \rangle_r. \label{eqn:quadvariation}
\end{align}
Now noting that we will take $\Ex_{t,x}^{t, T}$ in the end, and using It\^o's formula to deduce the representation
\begin{align}
\label{eqn:urep}
\begin{split}
	u_x(t_{i+1}, X_r) \sigma(r, X_r, V_r) &= u_x(t_{i+1}, X_t) \sigma(r, X_t, V_r) + \int_t^r \partial_x (u_x(t_{i+1}, X_\theta) \sigma(r, X_\theta, V_r)) \dd X_\theta \\&
	\quad + \frac{1}{2} \int_t^r \partial_{xx} (u_x(t_{i+1}, X_\theta) \sigma(r, X_\theta, V_r) ) \dd \langle X, X \rangle_\theta
\end{split}
\end{align} 
we can compute the quadratic covariation term \cref{eqn:quadvariation} further:
\begin{align}
	\Ex_{t, x}^{t, T} \int_{t_i}^{t_{i+1}} \dd \langle u_x(t_{i+1}, X_\cdot) \rho_\cdot \sigma(\cdot , X_\cdot , V_\cdot) , B_\cdot \rangle_r &=\int_{t_i}^{t_{i+1}}  \Ex_{t, x}^{t, T} u_x(t_{i+1}, x) \rho_r \dd \langle \sigma(\cdot, x, V_\cdot), B_\cdot \rangle_r \label{eqn:quadvariationcomp1} \\
	&= \int_{t_i}^{t_{i+1}}  \Ex_{t, x}^{t, T} u_x(t_{i+1}, x) \rho_r \sigma_y (r, x, V_r) \beta(r, V_r) \dd r \label{eqn:quadvariationcomp2} \\
	&= \Ex_{t, x}^{t, T} \int_{t_i}^{t_{i+1}} u_x(t_{i+1}, x) \rho_r \sigma_y (r, x, V_r) \beta(r, V_r) \dd r \nonumber
\end{align}
where all the preceding equalities are understood up to some higher-order negligible terms (namely, $\littleo{\Delta t}$). Moreover, \cref{eqn:quadvariationcomp1} is true by substitution of \cref{eqn:urep}, and \cref{eqn:quadvariationcomp2} is obtained through additional use of It\^o's formula on $r \mapsto \sigma(r, x, V_r)$. Thus we obtain 
\begin{align*}
	\Ex_{t, x}^{t, T} \left [\chi_i \right ] &= \Ex_{t, x}^{t, T} \int_{t_i}^{t_{i+1}} \left (\LL_r^x - \CC_r^x \right )u(t_{i+1}, x) \dd r + \Ex_{t, x}^{t, T} \int_{t_i}^{t_{i+1}} \BB_r^x u(t_{i+1}, x) \bd B_r.
\end{align*}
%----------tau_i
The term $\tau_i$ is easy to handle, we simply use the SPDE \cref{eqn:spdeintro}, as $\tau_i = u(t_{i+1}, X_{t_i}) - u(t_i, X_{t_i})$, yielding
\begin{align*}
	 \Ex_{t, x}^{t, T} [\tau_i] &= - \Ex_{t, x}^{t, T} \int_{t_i}^{t_{i+1}}\left (\LL^{X_{t_i}}_r - \CC^{X_{t_i}}_r \right ) u(r,X_{t_i})\dd r - \Ex_{t, x}^{t, T}\int_{t_i}^{t_{i+1}}\BB^{X_{t_i}}_r u(r,X_{t_i}) \bd B_r.
\end{align*}
Reformulating \cref{eqn:telescoping1}, we deduce that our goal is to show
\begin{align*}
	\Ex^{t, T}_{t,x} \left [ \sum_{i = 0}^{n-1} \chi_i + \tau_i \right ] \longrightarrow 0
\end{align*}
in $L^1(\Qro_{t, x})$ as $n \to \infty$. Hence, we recognise that the choice of SPDE \cref{eqn:spdeintro} is correct (although, see \Cref{remark:informalSPDE} below). Essentially, the SPDE \cref{eqn:spdeintro} is chosen so as to ensure that the terms $\tau_i$ and $\chi_i$ are more or less the same but with opposite sign.
\begin{remark}
\label{remark:informalSPDE}
We stress that the above derivation is informal. There are a number of technicalities that are not addressed, most importantly, the above SPDE \cref{eqn:spdeintro} is not entirely correct as it is missing a correction term in the drift; this is due to the fact that the backward stochastic integral that appears in it is not well-defined in the It\^o sense. Ergo, the intention of this article is to address and formalise the above argument. Despite this, it should be remarked that the desired SPDE for numerical applications is in fact the one just derived. Roughly speaking, this is due to matters of existence of stochastic integrals not being important when time is discretised, and thus the previously mentioned correction term in the drift formally cancels out with a term in the driving noise. Indeed, \cref{eqn:spdeintro} is the one we use in order to numerically price a European put option using our mixed Monte-Carlo PDE method in \Cref{sec:numerics}.
\end{remark}

%-------------------------------------------------------------------------Preliminaries-------------------------------------------------------------------------------------
%-----------------------------------------------------------------------------------------------------------------------------------------------------------------------------

\section{Preliminaries}
\label{sec:preliminaries}
\noindent

\noindent We will utilise the following notation and terminology throughout this article. 
For functions $f, g$ with the same domain and codomain, we will often suppress the argument of all functions except the last when writing products. For example, $fg(x, y) \equiv f(x, y) g(x, y)$. Sometimes subscripts will denote a partial derivative of a function, for example, $f_x(x, y)  \equiv \partial_x f(x, y)$.
%--------Notation: Stochastic analysis
Let $\zeta$ be an arbitrary stochastic process. The following are different notations for the same object: $\Ex[f(\zeta_T) | \zeta_t =x] \equiv \Ex_{t,x}[f(\zeta_T)].$
Specifically, this means that the expectation is taken w.r.t. $\Qro_{t,x}(\cdot) := \Qro( \cdot | \zeta_t = x)$. We will denote by $\Delta \zeta_i := \zeta_{t_{i+1}} - \zeta_{t_i}$ the forward difference of $\zeta$ over some partition of $[0,T]$.

In the rest of the article we assume that all filtrations satisfy the usual conditions. For a forward filtration, this means it is right continuous and the initial element has been augmented by null sets, whereas in the case of a backward filtration, this means that it is left continuous and the terminal element has been augmented by null sets. The following notation will be used for a variety of specific $\sigma$-algebras.

\begin{enumerate}[label = (\roman*), ref = \roman*]

\item $\FF_{s,t}^\zeta := \sigma(\zeta_v - \zeta_u, s \leq u < v \leq t)$ denotes the $\sigma$-algebra generated by the \emph{increments} of $\zeta$ over the interval $[s,t]$.

\item $\bar \FF_{s,t}^\zeta := \sigma(\zeta_u, s \leq u \leq t)$ denotes the $\sigma$-algebra generated by the \emph{path} of $\zeta$ over the interval $[s,t]$. It is then clear that $\bar \FF_{s,t}^\zeta = \FF_{s,t}^\zeta \vee \sigma(\zeta_{t'})$, where $t' \in [s, t]$, i.e., the path over $[s,t]$ is equal to the increments over $[s,t]$ `plus' a point of $\zeta$ on $[s,t]$.

\item Given $\zeta_0$ is constant, we will write $\FF_t^\zeta \equiv \bar \FF^\zeta_{0,t} = \FF_{0,t}^{\zeta}$, which is the $\sigma$-algebra corresponding to the natural filtration of $\zeta$.

\end{enumerate}
We stress that there is a subtle distinction between the increments $\sigma$-algebra $\FF^\zeta_{s,t}$ and path $\sigma$-algebra $\bar \FF^\zeta_{s,t}$. The following remark is a simple example which illustrates this.
%-------Prop: Not a BBM in enlarged filtration
\begin{remark}
Let $Z$ be a standard Brownian motion w.r.t. its natural filtration $(\FF_t^Z)_{\tinT}$. Define $\tilde Z_t = Z_t - Z_T$. Then $\tilde Z$ is a backward Brownian motion in $(\FF_{t,T}^Z)_{\tinT}$. However, it is not a backward Brownian motion in $(\bar \FF_{t,T}^Z)_{\tinT}$. It is easy to see this as 
\begin{align*}
	\Ex[\tilde Z_0 | \bar \FF_{t, T}^Z ] = \Ex[Z_0 - Z_T | \FF_{t, T}^Z, Z_T ] = - Z_T = \tilde Z_0 \neq \tilde Z_t. 
\end{align*} 
Hence, $\tilde Z$ is not a backward martingale in $(\bar \FF_{t,T}^Z)_{\tinT}$, and thus not a backward Brownian motion.\footnote{In fact, we have that 
\begin{align*}
	\Ex[\tilde Z_s | \bar \FF_{t, T}^Z ] = \tilde Z_t - \Ex \left [ \int_s^t \frac{Z_r}{r} \dd r | \bar \FF_{t,T}^Z \right ]. 
\end{align*} 
This can be seen by adapting the classical Brownian bridge example from initial enlargement of filtration theory. Namely, $\tilde Z$ remains a semimartingale in the backward filtration $(\bar \FF_{t, T}^Z)_{\tinT}$ and moreover possesses the decomposition $\tilde Z_t = \bar Z_t - \int_t^T \frac{Z_r}{r} \dd r$, with $\bar Z$ being a backward Brownian motion in $(\bar \FF_{t, T}^Z)_{\tinT}$. See \citep[][Chapter 5.9]{jeanblanc2009mathematical} for further information.
}
\end{remark}

%-------Notation: Function spaces
Let $(S, \SS)$ be a measurable space, where $S$ is a real, separable Hilbert space with inner product $\langle \cdot, \cdot \rangle_S$ and induced norm $\|\cdot \|_S := \sqrt{\langle \cdot, \cdot \rangle_S}$. 
In the following, $U$ denotes an open subset of $\reals^n$. The space $C(U ; S)$ consists of functions $\psi: U \to S$ which are continuous. The space $C^k(U ; S)$ consists of $k$-times (strongly) differentiable functions $\psi: U \to S$, whose $k$-th derivative is continuous. Spaces $C_c(\dots)$ and $C_c^k(\cdots)$ will denote the subspace of $C(\cdots)$ and $C^k(\cdots)$ containing functions with compact support respectively, whereas $C_b(\cdots)$ and $C_b^k(\cdots)$ will denote the subspace of $C(\cdots)$ and $C^k(\cdots)$ containing functions which have bounded partial derivatives up to order $k$ respectively. We will write $\mathbb{B}(X, Y)$ to denote the space of bounded linear operators from $X$ to $Y$. Let $(X, \XX, \mu)$ be a measure space. Integration of measurable functions $\psi: (X, \XX ) \to (S, \SS)$ w.r.t. $\mu$ is understood in the Bochner sense. Consider the norm
\begin{align*}
 	\|\psi \|_{L^p((X, \XX, \mu);S)} := \begin{cases}
 	 \left (\int_X \| \psi (x) \|_S^p \mu(\dd x)\right )^{1/p}, &1 \leq p < \infty,  \\
 	 \esssup_{x \in X} \| \psi(x) \|_S, &p = \infty.
 	 \end{cases}
\end{align*}
Then
\begin{align*}
 	L^p((X, \XX, \mu) ;S) := \{ \psi : \|\psi \|_{L^p((X, \XX, \mu); S)} < \infty \}
\end{align*}
is a Banach space for $1 \leq p \leq \infty$, where functions in this space are identified $\mu$ a.e. Moreover, $L^2((X, \XX, \mu) ; S)$ is a Hilbert space with inner product $\langle \psi_1, \psi_2 \rangle_{L^2((X, \XX, \mu) ; S)} :=  \int_X \langle \psi_1(x) , \psi_2(x) \rangle_{S} \mu(\dd x)$. Often when writing $L^p$ spaces, only some of the arguments of the corresponding measure space will be significant, and thus we may omit some arguments for notational convenience. For example, the space $L^p((X, \XX, \mu); \SS)$ could be written as $L^p(\mu ; \SS)$, or $L^p(X)$. This notation will carry over to the inner products and norms.

Let $k \in \naturals$ and $1 \leq p \leq \infty$. We denote by $W^{k, p}(U)$ the Sobolev space given by
\begin{align*}
	W^{k, p}(U) := \{ \psi : U \to \reals \mid  \partial^{\alpha} \psi \in L^p(U ; \reals), \text{ for all }0 \leq |\alpha| \leq k \},
\end{align*}
where we utilise the multi-index notation $\partial^{\alpha} \psi := \frac{\partial^{|\alpha|} \psi}{\partial x_1^{\alpha_1} \cdots \partial x_n^{\alpha_n}}$, with $\alpha \in \naturals_0^n$ and $|\alpha| := \alpha_1 + \dots + \alpha_n$. Moreover, $W^{k, p}(U)$ is a Banach space with norm 
\begin{align*}
	\| \psi \|_{W^{k,p}(U) } := \begin{cases}
		 \left (\sum_{|\alpha| \leq k} \int_{U} | \partial^\alpha \psi(x) |^p \dd x \right )^{1/p}, &1 \leq p < \infty, \\
		 \sum_{|\alpha| \leq k } \esssup_{x \in U} |\partial^{\alpha} \psi(x)|, &p = \infty.
		 \end{cases}
\end{align*}
We will write $H^{k}(U) := W^{k, 2}(U)$, which is a Hilbert space with inner product 
\begin{align*}
	\langle \psi_1, \psi_2 \rangle_{H^k(U)} := \sum_{|\alpha| \leq k} \int_{U} \partial^\alpha \psi_1(x) \partial^\alpha \psi_2(x) \dd x.
\end{align*}
We will make use of the following common abuse of notation. When $U$ is an open interval, e.g., $(a,b)$ we will write $C (a,b ; S) \equiv C((a,b) ; S)$, $L^p(a,b ; S) \equiv L^p((a,b) ; S)$, and so forth. We will often omit the codomain when it is clear, e.g., $C^k(\reals^n) \equiv C^k(\reals^n ; \reals)$, $L^p(\reals^n) \equiv L^p(\reals^n ; \reals)$, and so forth.

%-----------------------------------------------------------------Model framework-------------------------------------------------------------------------------------
\subsection{Model framework}
Fix a finite time horizon $T > 0$. Let $W$ and $B$ be one-dimensional Brownian motions on a complete probability space $(\Omega, \FF, \Qro)$, with deterministic time-dependent instantaneous correlation $(\rho_t)_{\tinT}$. In the following, we consider the diffusion process $(X, V)$ taking values in $\reals^2$ and given by the (forward) system
\begin{align}
	\dd X_t &= \mu(t, X_t, V_t) \dd t + \sigma(t, X_t, V_t) \dd W_t, \label{eqn:system1X}\\
	\dd V_t &= \alpha(t, V_t) \dd t + \beta(t, V_t) \dd B_t, \label{eqn:system1V} \\  
	\dd \langle W, B \rangle_t &= \rho_t \dd t. \nonumber
\end{align}
Here $\mu, \sigma: [0,T] \times \reals \times \reals \to \reals$ and $\alpha, \beta : [0,T] \times \reals \to \reals$ are Borel measurable and deterministic. The system \crefrange{eqn:system1X}{eqn:system1V} can be rewritten as
\begin{align}
	\dd X_t &= \mu(t, X_t, V_t) \dd t + \rho_t \sigma(t, X_t, V_t) \dd B_t + \varrho_t \sigma(t, X_t, V_t) \dd \hat B_t, \label{eqn:system2X} \\
	\dd V_t &= \alpha(t, V_t) \dd t + \beta(t, V_t) \dd B_t \label{eqn:system2V}
\end{align}
where $\hat B$ is a one-dimensional Brownian motion independent of $B$, and $\varrho_t := \sqrt{1 - \rho_t^2}$. Here $w := (B, \hat B)$ is a standard two-dimensional Brownian motion, and we denote its natural filtration by $(\FF_t^w)_{\tinT}$, which satisfies the usual conditions.

\begin{remark}
To simplify ideas and reduce notation, we will be content with remaining in the two-dimensional setting. Later on in \Cref{sec:multivariablesetting} we will tackle the general multivariable setting.
\end{remark}

We will enforce the following standard assumption throughout the rest of this article. Its purpose is to guarantee the existence of a pathwise unique strong solution for the system \crefrange{eqn:system2X}{eqn:system2V} which does not blow up in finite time. It is a mixture of the usual It\^o style existence and uniqueness criteria for SDEs, as well as the Yamada-Watanabe condition (see \citep[][Theorem 1]{yamada1971uniqueness}), the latter of which can only be applied to $V$ as it is decoupled from $X$. 

%--------Assumption: SDE
\begin{assumption}
\label{ass:sde}
\noindent
\begin{enumerate}[label = (A\arabic*), ref = A\arabic*]
\item $(x, y) \mapsto \mu(t, x, y)$ and $(x, y) \mapsto \sigma(t, x,  y)$ are locally Lipschitz continuous, uniformly in $t$. 
\item $|\mu(t, x, y)| + |\sigma(t, x, y)| \leq C(1 + |(x, y)|)$, uniformly in $t$.
\item There exists a weak solution  $V$ to \cref{eqn:system2V}. Moreover, there exists non-decreasing functions $\kappa, \gamma: (0, \infty) \to (0, \infty)$ where in addition, $\kappa$ is concave with $\lim_{\vep \downarrow 0} \int_\vep^1 1/\kappa(u) \dd u = \lim_{\vep \downarrow 0} \int_\vep^1 1/\gamma^2(u) \dd u = +\infty$ such that for all $y, y'$ we have $|\alpha(t, y) - \alpha(t, y') | \leq \kappa(y - y')$ and $|\beta(t, y) - \beta(t, y') | \leq \gamma(y - y')$, uniformly in $t$. \label{item:yamada}
\item $|\alpha(t, y)| + |\beta(t, y)| \leq C(1 + |y|)$, uniformly in $t$.
\end{enumerate}
\end{assumption}
%-------Backward comment
In the rest of the article, we will encounter a so-called backward stochastic integral, which shall be understood in the sense of It\^o. Intuitively, a backward stochastic integral ought to possess the following traits. First, its integrand is adapted to a backward filtration generated by the integrator. Indeed, inverting the flow of time should result in the time flow of information being inverted; i.e., our filtration should evolve backwards in time. Secondly, the construction of the integral is done backward, hence, the Riemann sums utilise backward differencing. In other words, this means that the right end point of the integrand is chosen in the Riemann sums. This motivates the following definition.

%-------Defn: Backward stochastic integral
\begin{definition}[Backward stochastic integral]
\label{defn:backstoch}
Let $Z$ be a backward Brownian motion in a backward filtration $(\GG_{t,T})_{\tinT}$. Let $\zeta$ be adapted to $(\GG_{t,T})_{\tinT}$. The backward stochastic integral of $\zeta$ against $Z$ is defined as
\begin{align*}
	\int_t^T \zeta_r \bd Z_r := \lim_{\delta_n \downarrow 0} \sum_{i=0}^{n-1} \zeta_{t^{(n)}_{i+1}} (Z_{t^{(n)}_{i+1}} - Z_{t^{(n)}_i})
\end{align*}
where $\delta_n := \sup_i (t^{(n)}_{i+1} - t_i^{(n)})$ corresponds to the mesh of the $n$-th partition $\{t = t^{(n)}_0 < \dots < t^{(n)}_{n-1} < t^{(n)}_n = T\}$, and the limit is in probability.
\end{definition}

The existence of the backward stochastic integral can be proved by simply proceeding with the usual construction of the (forward) It\^o integral.

%-------Remark: integration against \tilde B abuse of notation
\begin{remark}
\label{remark:Btildeint}
Let $\tilde B_t := B_t - B_T$, where $B$ refers to the forward Brownian motion driving $V$ from \cref{eqn:system2V}. Then $\tilde B$ generates the backward filtration $(\FF_{t,T}^B)_{\tinT}$, i.e., the backward filtration generated by the increments of $B$ on $[t,T]$. Moreover, $\tilde B$ is a standard backward Brownian motion w.r.t. $(\FF_{t,T}^B)_{\tinT}$. Let $\zeta$ be adapted to $(\FF_{t,T}^B)_{\tinT}$. Then we will use the following abuse of notation:
\begin{align*}
	\int_t^T \zeta_r \bd B_r := \int_t^T \zeta_r \bd \tilde B_r
\end{align*}
where the RHS exists as a backward stochastic integral in the sense of \Cref{defn:backstoch}. Note that this is an abuse of notation since $\tilde B$ is a standard backward Brownian motion relative to $(\FF^B_{t,T})_{\tinT}$, not $B$. 
\end{remark}

Define $\bar \FF_{t,T}^{V,B} := \FF_{t,T}^B \vee \sigma(V_t)$, the $\sigma$-algebra generated by the increments of $B$ on $[t,T]$ and the random variable $V_t$, these processes being defined in \crefrange{eqn:system2X}{eqn:system2V}. Note that also, $\bar \FF_{t,T}^{V,B} =  \FF_{t,T}^B \vee \sigma(V_T)$.

%-------Remark: some backward stoch integrals do not exist
\begin{remark}
Let $\eta$ be adapted to $(\bar \FF_{t,T}^B)_{\tinT}$ and $\xi$ be adapted to $(\bar \FF_{t,T}^{V, B})_{\tinT}$. From \Cref{remark:Btildeint}, $\tilde B_t := B_t - B_T$ is a standard backward Brownian motion relative to $(\FF_{t,T}^B)_{\tinT}$. Then the backward stochastic integrals 
\begin{align*}
	\int_t^T \eta_r \bd \tilde B_r    \quad \text{ and } \quad 	\int_t^T \xi_r \bd \tilde B_r
\end{align*}
do not exist in the sense of It\^o, i.e., in the sense of \Cref{defn:backstoch}. This can be seen by noting that the It\^o isometry fails when attempting their construction in the corresponding backward filtrations.
\end{remark}

%---------\mrB
Suppose that $V_t$ possesses a density $p(t,y)$ w.r.t. Lebesgue measure. That is, $\Qro(V_t \in A) = \int_A p(t,y) \dd y$ for any Borel set $A$ in $\reals$. Define the process 
\begin{align}
	\mrB_t := B_t - B_T - \int_t^T\frac{\partial_y (p(r,V_r) \beta(r,V_r))}{p(r,V_r)} \dd r \label{eqn:mrB}
\end{align} 
where the integrand is taken to be zero if ever $p$ is zero. To ensure $\mrB$ is well-defined, we will require the following assumption, which we will enforce from here on in:

%---------Assumption: \mrB welldefined
\begin{assumption}
\label{ass:mrB}
\noindent
\begin{enumerate}[label = (B\arabic*), ref = B\arabic*]
 \item The density of $V_0$,  $p_0(y) \equiv p(0,y)$ satisfies $\int_{\reals} \frac{p^2_0(y)}{1 + |y|^k} \dd y < \infty$ for some $k \in \naturals$.
 \item  $(\partial^2_y \beta^2) \in L^{\infty}([0, T] \times \reals ; \reals)$.
\end{enumerate}
\end{assumption}
Hence by \Cref{thm:definingmrB} with $D = 1$, $\mrB$ is a backward Brownian motion in $(\bar \FF_{t,T}^{V,B})_{\tinT}$. 

The following remark quantifies how utilising $\mrB$ vs $\tilde B$ as the stochastic integrator affects calculations.

%-------Remark: backward integration against mrB
\begin{remark}
\label{remark:mrBvsB}
Let $\xi$ be adapted to $(\bar \FF_{t,T}^{V, B})_{\tinT}$. Then the backward stochastic integral 
\begin{align*}
	\int_t^T \xi_r \bd \mrB_r
\end{align*}
exists in the sense of \Cref{defn:backstoch}. However, supposing $\xi$ is simple on some partition $\{t = t_0 < \dots < t_{n-1} < t_n = T\}$, we have
\begin{align*}
	\int_t^T \xi_r \bd \mrB_r = \sum_{i=0}^{n-1} \xi_{t_{i+1}} \Delta \mrB_i \neq \sum_{i=0}^{n-1} \xi_{t_{i+1}} \Delta B_i.
\end{align*}
Thus, if for argument's sake we supposed $\int_t^T \xi_r \bd \tilde B_r$ existed, then $\int_t^T \xi_r \bd \mrB_r$ would not coincide with it. In fact, we have
\begin{align*}
	\Delta \mrB_i = \Delta B_i + \int_{t_i}^{t_{i+1}} \frac{\partial_y (p(r,V_r) \beta(r,V_r))}{p(r,V_r)} \dd r.
\end{align*}
Hence despite it being It\^o sense ill-posed, we can informally write an expression for $\int_t^T \xi_r \bd \tilde B_r$, namely
\begin{align*}
	\int_t^T \xi_r \bd \tilde B_r \stackrel{\text{informal}}{=} \int_t^T \xi_r \bd \mrB_r - \int_t^T   \xi_r \frac{\partial_y (p(r,V_r) \beta(r,V_r))}{p(r,V_r)} \dd r.
\end{align*}
\end{remark}

%------------------------------------------------------------------------The SPDE---------------------------------------------------------------------------------------
\subsection{The SPDE}

%-------The SPDE
The main focus of this article will be the following (backward) SPDE:
\begin{align}
\begin{split}
	-\dd u (t, x) &= \left (\LL^x_t - \CC^x_t - \frac{\partial_y (p(t,V_t) \beta(t,V_t))}{p(t,V_t)} \BB_t^x \right ) u(t,x)\dd t + \BB^x_t u(t,x) \bd \mrB_t, \label{eqn:spdewellposed} \\
	 u(T,x) &= \varphi(x),
\end{split}
\end{align}
where we have the following family of (stochastic) differential operators indexed by $t \in [0, T]$,
\begin{align}
	\LL^x_t &:= \frac{1}{2} \sigma^2(t,x,V_t) \partial_x^2 + \mu (t, x, V_t) \partial_x, \label{eqn:stochdiffoperatorL}\\
	\BB^x_t &:=  \rho_t \sigma(t,x,V_t) \partial_x, \label{eqn:stochdiffoperatorB} \\
	\CC^x_t &:= \rho_t \beta(t,V_t) \sigma_y(t, x, V_t) \partial_x. \label{eqn:stochdiffoperatorC}
\end{align}

%-------From the perspective of Mathematical finance
From the perspective of mathematical finance, the purpose of studying the SPDE \cref{eqn:spdewellposed} is the following. Suppose that $(\intr_t)_{t \in [0,T]}$ is the deterministic interest rate, and assume that $\Qro$ is a chosen risk-neutral measure. Let $H$ be the price of a European style derivative on $X$, meaning its payoff $\varphi$ only depends on the terminal value of $X$. Specifically
\begin{align*}
	H_t &=e^{-\int_t^T \intr_r \dd r} \, \Ex\big [\varphi(X_T) | \FF^w_t \big ].
\end{align*}
Recall $\bar \FF_{t, T}^{V, B} = \FF_{t, T}^B \vee \sigma(V_t)$. Let\footnote{At this point one will note that the $\sigma$-algebra $\GG_{t,T}$ from \Cref{sec:introduction} is $\bar \FF_{t,T}^{V,B}$.}
\begin{align}
	\bar u(t, x) := \Ex[ \varphi(X_T) |X_t = x, \bar \FF^{V,B}_{t,T}]. \label{vexpression}
\end{align}
Then
\begin{align*}
	H_t &\stackrel{\text{Markov}}{=}  e^{-\int_t^T \intr_r \dd r} \, \Ex \big [\varphi(X_T) | X_t, V_t \big ] = e^{-\int_t^T \intr_r \dd r} \, \Ex \big [\Ex [ \varphi(X_T) | X_t, \bar \FF_{t,T}^{V,B} ]| X_t, V_t \big ] \\ 
	& \quad= e^{-\int_t^T \intr_r \dd r} \, \Ex \big [\bar u(t, X_t)| X_t, V_t \big ].
\end{align*}
In particular, 
\begin{align*}
	H_0 = e^{-\int_0^T \intr_r \dd r} \, \Ex \big [\Ex[\varphi(X_T) | \bar \FF_{0,T}^{V,B} ] \big ] = e^{-\int_0^T \intr_r \dd r} \, \Ex \big [ \bar u(0, x) \big ].
\end{align*}
We prove that $\bar u(t,x)$ solves the SPDE \cref{eqn:spdewellposed} in \Cref{thm:conditionalfeynmanwellposed}, thereby establishing a connection between derivative pricing and SPDE theory. This result can be utilised for the pricing of American style derivatives through Least Square Monte-Carlo methods by applying it to the continuation value, as well as in other areas of mathematical finance. These applications will be studied in forthcoming articles. The focus of this article however, will be on developing a rigorous foundation for the theory.

%-------Remark: Variational formulation
\begin{remark}[Variational formulation]
\label{remark:variational}
A solution to the SPDE \cref{eqn:spdewellposed} is to be understood through its variational formulation.\footnote{Precisely, weak in the PDE sense, and strong in the stochastic analysis sense.} To do so we first multiply $\LL_t^x u$ by a test function $v \in H^1(\reals)$ and integrate, thereby obtaining the following expression via integration by parts:
\begin{align*}
	\int_\reals (\LL_t^x u) v(x) \dd x = - \frac{1}{2} \int_\reals \sigma^2(t, x, V_t) u_x v_x(x) \dd x + \int_\reals \left (\mu(t, x, V_t) - \frac{1}{2} \partial_x(\sigma^2(t, x, V_t)) \right ) u_x v(x) \dd x.
\end{align*}
Thus as is standard, $\LL_t^x$ implicitly defines a bilinear form on $H^1(\reals) \times H^1(\reals)$ for almost all $\omega \in \Omega$. Hence, for almost all $\omega \in \Omega$, it makes sense to think of $\LL_t$ as a family of bounded linear operators $(\LL_t)_{t \in [0, T]}$ with $\LL: [0, T] \to \mathbb{B}(H^1(\reals), H^{-1}(\reals))$, so that the natural pairing is given by
\begin{align*}
	\langle \LL_t u, v \rangle =  - \frac{1}{2} \int_\reals \sigma^2(t, x, V_t) u_x v_x(x)  \dd x + \int_\reals \left (\mu(t, x, V_t) - \frac{1}{2} \partial_x(\sigma^2(t, x, V_t)) \right ) u_x v(x) \dd x,
\end{align*}
for any $u, v \in H^1(\reals)$. Then, writing $u(t) \equiv u(t, \cdot)$, we get the following variational formulation for the SPDE \cref{eqn:spdewellposed}:
\begin{align*}
	- \dd \langle u(t), v \rangle_{L^2(\reals)} &= \left ( \langle \LL_t u(t), v \rangle  - \langle \CC_t u(t), v \rangle_{L^2(\reals)} -  \frac{\partial_y (p(t,V_t) \beta(t,V_t))}{p(t,V_t)} \langle  \BB_t u(t) , v \rangle_{L^2(\reals)} \right ) \dd t \\
	&\quad+ \langle \BB_t u(t), v \rangle_{L^2(\reals)} \bd \mrB_t, \\
	\langle u(T), v \rangle_{L^2(\reals)} &=  \langle \varphi, v \rangle_{L^2(\reals)},
\end{align*}
for any $v \in H^1(\reals)$. 
\end{remark}

%-------Assumption: SPDE
In order to ensure our main results pertaining to the SPDE \cref{eqn:spdewellposed} are valid, we will here on in enforce the following assumption. 
\begin{assumption}
\label{ass:spde}
\noindent
\begin{enumerate}[label = (C\arabic*), ref = C\arabic*]
\item  $\varphi \in C_c^1(\reals; \reals)$. \label{ass:C1}
\item $\mu, \sigma \in L^{\infty}([0, T] \times \reals \times \reals; \reals)$ and $\alpha, \beta \in L^{\infty}([0, T] \times \reals; \reals)$. \label{ass:C2}
\item $\partial_x \sigma, \partial_y \sigma \in L^{\infty}([0, T] \times \reals \times \reals ; \reals)$ and are continuous in $(x,y)$ on compacts of $[0,T] \times \reals \times \reals$, uniformly in $t$. \label{ass:C3}
\item $\sigma^2(t,x,y) \geq C $ for some constant $C>0$, uniformly in $(t,x,y)$. \label{ass:C4}
\end{enumerate}
\end{assumption}

Lastly, we will need to make the following assumption in order to control the speed of growth of the density of $V_r$.
%-------Assumption: Ratio density
\begin{assumption}
\label{ass:ratio_density}
Recall $p(r, y)$ is the density of $V_r$.
\begin{align*}
	\left | \frac{\partial_y(p(r, y)\beta(r, y))}{p(r, y)} \right | \leq C \left ( \frac{|y|^{p_1}}{r^{q_1}} + \frac{|y|^{p_2}}{r^{q_2}}\right ),
\end{align*}
where $p_i \geq 0, q_i \in \reals$ and $p_i = 0$ implies $q_i \leq 0$, for $i = 1, 2$.
\end{assumption}

%-------------------------------------------------------------------------Main results-----------------------------------------------------------------------------------
%--------------------------------------------------------------------------------------------------------------------------------------------------------------------------
\section{Main results}
\label{sec:mainresults}
\noindent In this section, we provide the main results, which we will then prove in \Cref{sec:proofs}. We reiterate that in the following results, \Crefrange{ass:sde}{ass:ratio_density} are being enforced.

The following theorem is an adaptation of \citep[][Theorem 6.1]{pardoux1982equations}.
%-------Theorem: The well-posed SPDE solution exists
\begin{theorem}
\label{thm:spdeexistence} 
There exists a unique solution $u(t,x)$ to the SPDE \cref{eqn:spdewellposed}, adapted to $(\bar \FF_{t,T}^{V, B})_{\tinT}$. Moreover, $t \mapsto u(t,x)$ belongs to $L^2(\vep, T ; H^1(\reals)) \cap C([\vep, T]; L^2(\reals))$ for all $\vep > 0$, $\Qro$ a.s.
\end{theorem}

The following results pertain to the conditional Feynman-Kac formula, and these are extensions of Proposition 6.4 and Theorem 6.5 in \citep{pardoux1982equations}. Our innovation comes from the fact that we are required to condition on the $\sigma$-algebra $\bar \FF_{t,T}^{V,B}$ rather than $\bar \FF_{t,T}^B$ or $\bar \FF_{t,T}^V$, thereby requiring the use of the backward Brownian motion $\mrB$ from the enlarged filtration $(\bar \FF_{t,T}^{V,B})_{\tinT}$ as the backward stochastic integrator. As a consequence of this, enforcing \Cref{ass:ratio_density} is critical.

%-------Proposition: Conditional Feynman-Kac formula
\begin{proposition}
\label{prop:conditionalfeynmanwellposedprop} 
Let $u(t,x)$ be the unique $(\bar \FF_{t,T}^{V, B})_{\tinT}$-adapted solution to the SPDE \cref{eqn:spdewellposed}. Assume in addition to \Crefrange{ass:sde}{ass:ratio_density} that:
\begin{enumerate}[label = (E\arabic*), ref = E\arabic*]
\item $\varphi \in C_c^{\infty}(\reals;\reals)$. \label{ass:extrareg1}
\item $\mu, \sigma, \alpha, \beta$ possess partial derivatives of all orders in time and space, which in addition, are all bounded, and continuous in space uniformly in $t$ on compacts of $[0, T] \times \reals^2$ for $\mu, \sigma,$ and $[0, T] \times \reals$ for $\alpha, \beta$. \label{ass:extrareg2}
\end{enumerate}
Then for all $t \in (0, T]$ and $x \in \reals$, $u(t,x)$ admits the representation
\begin{align*}
	u(t, x) = \Ex \big [ \varphi(X_T) | X_t = x, \bar \FF_{t,T}^{V,B}]
\end{align*}
$ \Qro$ a.s.
\end{proposition}

The previous proposition will be utilised to prove the following theorem, which is our main result.

%-------Theorem: Conditional Feynman-Kac formula
\begin{theorem}
\label{thm:conditionalfeynmanwellposed} 
Let $u(t,x)$ be the unique $(\bar \FF_{t,T}^{V, B})_{\tinT}$-adapted solution to the SPDE \cref{eqn:spdewellposed}. Then for all $t \in (0, T]$, $u(t,x)$ admits the representation
\begin{align*}
	u(t, x) = \Ex \big [ \varphi(X_T) | X_t = x, \bar \FF_{t,T}^{V,B}]
\end{align*}
$\dd x \times \dd \Qro$ a.e.
\end{theorem}

%-------Remark: Informal SPDE
\begin{remark}
\label{remark:spdeinformal}
As suggested in \Cref{sec:introduction}, the SPDE \cref{eqn:spdewellposed} can be restated in the \emph{informal} manner:
\begin{align}
\begin{split}
	-\dd u (t, x) &= \left (\LL^x_t - \CC^x_t \right ) u(t,x)\dd t + \BB^x_t u(t,x) \bd B_t, \label{eqn:spdeinformal} \\
	 u(T,x) &= \varphi(x).
\end{split}
\end{align}
However, the SPDE \cref{eqn:spdeinformal} is ill-posed (hence informal), as the backward stochastic integral in this expression is undefined in the It\^o sense. This is because the integrator is $B$, but the integrand, $\BB_t^x u(t,x)$, is $(\bar \FF_{t,T}^{V,B})_{\tinT}$-adapted, and thus It\^o's construction of stochastic integrals will not work. Specifically, the It\^o isometry fails when the integrand is not $(\FF_{t,T}^B)_{\tinT}$-adapted. To remedy this, we must use $\mrB$ as the integrator, which ends up adding a compensating term into the drift (see \Cref{remark:mrBvsB}), yielding the SPDE \cref{eqn:spdewellposed}. For this reason, from now on we may call \cref{eqn:spdewellposed} and \cref{eqn:spdeinformal} the `well-posed SPDE' and `informal SPDE' respectively. In short, there are two correction terms for the well-posed SPDE \cref{eqn:spdewellposed}:
\begin{enumerate}[label = (\arabic*), ref = \arabic*]
\item $\CC_t^x:$ this is a quadratic covariation term introduced due to `time-reversal' of the stochastic integral. This term is also present in the informal SPDE \cref{eqn:spdeinformal}. The intuition is the following: for a simple process $\zeta$ on $\{t = t_0 < \cdots < t_{n-1} < t_n = T\}$, we have 
\begin{align*}
	\sum_{i=0}^{n-1} \zeta_{t_i} \Delta B_i = \sum_{i=0}^{n-1} \zeta_{t_{i+1}} \Delta B_i + \sum_{i=0}^{n-1} \Delta \zeta_i \Delta B_i.
\end{align*}
The LHS is a forward differencing stochastic integral, whereas the RHS is a backward differencing stochastic integral plus a quadratic covariation term.
\item $\frac{\partial_y (p(t,V_t) \beta(t,V_t))}{p(t,V_t)}  \BB_t^x$: this is present in order to introduce $\mrB$ as the backward stochastic integrator, thereby ensuring existence of the stochastic integral (in the It\^o sense) and hence well-posedness of the SPDE, see \Cref{remark:mrBvsB}.
\end{enumerate}
However, it turns out that the informal SPDE \cref{eqn:spdeinformal} is the desired choice in numerical applications. This is because when one discretises time in order to numerically solve the SPDE, the formal and informal versions end up being equivalent, as there is no longer any danger of stochastic integrals being ill-posed. We refer the reader to \Cref{sec:numerics} for further details.

\end{remark}

%--------Remark: what happens at t = 0?
\begin{remark}
\label{remark:time0issue}
The conditional Feynman-Kac formula (\Cref{thm:conditionalfeynmanwellposed}) does not necessarily hold at $t = 0$, this being the case as \Cref{thm:spdeexistence} states that the well-posed SPDE \cref{eqn:spdewellposed} has a solution belonging to $L^2(\vep,T; H^1(\reals)) \cap C([\vep, T]; L^2(\reals))$, for all $\vep > 0$. Moreover, this issue occurs because we take into account the possibility of the distribution of $V_0$ being degenerate (and this is usually the case in applications). However, for the purposes of establishing a mixed Monte-Carlo PDE method, this is not a problem. 

To see this consider the following. For $s \geq 0$ denote by $\Qro_s$ the measure associated with the solution of the system \crefrange{eqn:system2X}{eqn:system2V} such that $\Qro_s(X_s = x_s, V_s = v_s) = 1$ for some deterministic $x_s, v_s$. Now assume $(X, V)$ is the solution of the system \crefrange{eqn:system2X}{eqn:system2V} under $\Qro_0$; this indeed means the distribution of $V_0$ is degenerate. Furthermore, for simplicity assume $\intr_t = 0$ a.e. on $[0, T]$. Let $\bar u(t, x)$ be given by \cref{vexpression}, where we stress that the conditional expectation in that expression is under $\Qro_0$. For $\delta > 0$, we consider the event $\{X_\delta = x_{\delta}, V_{\delta} = v_{\delta}\}$ for some deterministic $x_\delta, v_\delta$. We now consider the price of a derivative at time $t = \delta > 0$:
\begin{align*}
	H_{\delta} \rind{\{X_\delta = x_{\delta}, V_{\delta} = v_{\delta}\}} &= \Ex_0[\varphi(X_T)| X_{\delta}, V_{\delta}]  \rind{\{X_\delta = x_{\delta}, V_{\delta} = v_{\delta}\}} = \Ex_0[ \bar u(\delta, X_{\delta}) | X_\delta, V_\delta]  \rind{\{X_\delta = x_{\delta}, V_{\delta} = v_{\delta}\}} \\ &= \Ex_0[ \bar u(\delta, x_\delta) | X_\delta = x_\delta, V_\delta = v_\delta]  \rind{\{X_\delta = x_{\delta}, V_{\delta} = v_{\delta}\}}.
\end{align*}
The remarkable point here is that the density function $p(t, y)$ that enters into the well-posed SPDE \cref{eqn:spdewellposed} as well as in the definition of $\mrB$ (\cref{eqn:mrB}) is the one associated with the measure $\Qro_0$, not $\Qro_\delta$. Thus the troublesome point in the well-posed SPDE occurs at time $t = 0$, not $t = \delta$. Ergo, a mixed Monte-Carlo PDE method to simulate $H_\delta \rind{\{X_\delta = x_{\delta}, V_{\delta} = v_{\delta}\}}$ is to numerically solve the SPDE back to time $t = \delta$ to obtain i.i.d. copies of $\bar u(\delta, x_\delta)$ under $\Qro_\delta$, and then average over them.

We note that the same arguments apply if one simply considered shifting the time interval to $[-\tilde \delta, T]$ for some $\tilde \delta > 0$, and then used the preceding strategy to develop a mixed Monte-Carlo PDE method at time $t = 0$. Finally, knowing that a mixed Monte-Carlo PDE method can be established in the well-posed SPDE setting at $t = 0$, it is then legitimate to develop a mixed Monte-Carlo PDE method in the informal SPDE setting, which we indeed do in \Cref{sec:numerics}.

\end{remark}

%----------------------------------------------------------------------------Proofs-----------------------------------------------------------------------------
%------------------------------------------------------------------------------------------------------------------------------------------------------------------
\section{Proofs of main results}
\label{sec:proofs}

\noindent In this section, we provide the proofs of the main results from \Cref{sec:mainresults}. The strategies utilised in our proofs are similar to those considered in \citep{pardoux1982equations}. Our main innovation comes from the fact that we condition on $\bar \FF_{t, T}^{V, B}$ and thus the backward Brownian motion $\mrB$ defined in \cref{eqn:mrB} must be utilised as the stochastic integrator. In turn, this brings forth a number of non-trivial technicalities in the proofs. Thus, we will highlight aspects of the proofs where the consequences of $\mrB$ become apparent.

For the proofs in this section, we will need to discretise time. Consider a sequence of refining partitions $\PP_n := \{t = t_0^{(n)} < t_1^{(n)} < \cdots < t_{n-1}^{(n)} < t_n^{(n)}= T\}$ of $[t, T]$ where $n \in \naturals$. For brevity, we will usually write $t_i \equiv t_i^{(n)},$ unless the specific dependence on $n$ is required to avoid confusion. Let $ \Delta t \equiv t_{i+1} - t_i =  (T - t)/n$, i.e., each partition is uniform.

%---------------Difference scheme
Define the sequence $(u_i(x))_{i}$ through the following difference scheme:
\begin{align}
\begin{split}
	u_i(x) - u_{i+1}(x) &= \LLscr_i^xu_i(x) \Delta t - \CCscr_i^x u_{i+1}(x) \Delta t - \AAscr^x_i u_{i+1}(x) \Delta t 
\\&\quad+  \BBscr_i^x u_{i+1}(x) \Delta \mrB_i, \quad i = n-1, \dots, 0, \\
	u_n(x) &= \varphi(x),
	 \label{eqn:diffscheme}
\end{split}
\end{align}
where 
\begin{equation}
\begin{aligned}[c]
	\LLscr_i^x &:= \frac{1}{\Delta t} \int_{t_i}^{t_{i+1}} \LL_{r}^x |_{V_t = V_{t_i}} \dd r, 		&	\LL^x_t |_{V_t = V_{t_i}} &:= \frac{1}{2} \sigma^2(t,x,V_{t_i}) \partial_x^2 + \mu (t, x, V_{t_i}) \partial_x, \\
	\CCscr_i^x &:= \frac{1}{\Delta t} \int_{t_i}^{t_{i+1}} \CC_{r}^x |_{V_t = V_{t_i}} \dd r, 		&	\CC^x_t|_{V_t = V_{t_i}} &:= \rho_t \beta(t,V_{t_i}) \sigma_y(t, x, V_{t_i}) \partial_x, 
\\
	\BBscr_i^x &:=  \frac{1}{\Delta t} \int_{t_i}^{t_{i+1}} \BB_{r}^x |_{V_t = V_{t_{i+1}}} \dd r,  	 &	\BB^x_t|_{V_t = V_{t_{i+1}}} &:=  \rho_t \sigma(t,x,V_{t_{i+1}}) \partial_x, \\
	\AAscr_i^x &:= \frac{1}{\Delta t} \int_{t_i}^{t_{i+1}} \frac{\partial_y(p(r,V_r) \beta(r,V_r))}{p(r,V_r)} \BBscr_i^x  \dd r. & &
	\label{eqn:diffops}
\end{aligned}
\end{equation}
Hence, $\LLscr_i^x, \BBscr_i^x, \CCscr_i^x$ refer to the `average' versions of $\LL_t^x, \BB_t^x, \CC_t^x$ respectively. We will write $\LLscr_i \equiv \LLscr_i^\cdot, \BBscr_i \equiv \BBscr_i^\cdot, \CCscr_i \equiv \CCscr_i^\cdot, \AAscr_i \equiv \AAscr_i^\cdot$. Moreover, we will write $u_i \equiv u_i(\cdot) \in H^1(\reals)$ and $u(r) \equiv u(r, \cdot) \in H^1(\reals)$, where here we considered $\omega \in \Omega$ fixed. Thus, for each $i = n-1, \dots, 0$,  $u_i$ can be thought of as a $\bar \FF_{{t_i}, T}^{V,B}$ measurable random element, taking values in $H^1(\reals)$.

When constructing the difference scheme \cref{eqn:diffscheme}, we have simply discretised the SPDE \cref{eqn:spdewellposed}, however we have replaced the differential operators with their averages where the $V$ argument is frozen at either $t_i$ or $t_{i+1}$ as in \cref{eqn:diffops}. Furthermore. the operators $\LLscr_i^x, \BBscr_i^x, \CCscr_i^x, \AAscr_i^x$ act on either $u_i(x)$ or $u_{i+1}(x)$, this choice has been carefully decided and the reason will become apparent in the below proofs. Moreover, define
\begin{align}
	u^{(n)}(r, x) := \sum_{i=0}^{n-1} u_i(x) \gind{[t_i^{(n)}, t_{i+1}^{(n)})}{r} + u_n(x)\gind{\{t_n^{(n)} \}}{r} \label{eqn:usimple}
\end{align}
which is simple in $r$ on the partition $\PP_n$ for each $n$. We will write $u^{(n)}(r) \equiv u^{(n)}(r, \cdot) \in H^1(\reals)$, where here we considered $\omega \in \Omega$ fixed.

In the following proofs we will need to make use of some asymptotic notation. Consider an arbitrary random field $f$ whose mapping we will write as $f: \reals^2_+ \longrightarrow L^1(\Qro_{t, x})$.
\begin{itemize}
\item $f(r, s, \cdot) = \littleo{s - r}$ if 
\begin{align*}
	\frac{ \Ex_{t,x} |f(r, s, \cdot) |}{|s - r|} \longrightarrow 0 \text{  as  } |s - r| \to 0.
\end{align*} 

\item $f(r, s, \cdot) = \bigO{s -r}$ if there exists a constant $C > 0$ and a sufficiently small $r_0$ such that
\begin{align*}
	\Ex_{t,x} \left | f(r, s, \cdot) \right | \leq C |s - r|, \text{ when } |s - r| < r_0.
\end{align*}
\end{itemize}
The same notation will be used when considering the norm $\Ex | \cdot |$ rather than $\Ex_{t, x} | \cdot |$.

%---------------------------------------------------------------Proof of Theorem: SPDE existence-----------------------------------------------------------------------
\subsection*{Proof of \Cref{thm:spdeexistence}}
For the rest of the proof we will write $H^1 \equiv H^1(\reals)$ and $H^{-1} \equiv H^{-1}(\reals)$. Recall from \Cref{remark:variational} that $\langle \cdot, \cdot \rangle:H^{-1} \times H^1 \to \reals$ denotes the natural pairing of $H^{-1}$ and $H^1$ and moreover that $\LL_t$ can be interpreted as a family of bounded linear operators in $\mathbb{B}(H^1, H^{-1})$. Hence $I - \Delta t \LLscr_i$ is coercive for a sufficiently small $\Delta t$ by virtue of \Cref{ass:spde}, where $I$ denotes the identity operator. The idea is now classical; we would like that the sequence $(u^{(n)})_n$ defined in \cref{eqn:usimple} is bounded in $L^2(\Omega; L^2(t, T; H^1)) \cap L^2(\Omega; L^\infty(t, T; L^2(\reals)))$. This in turn will imply that there is a subsequence of  $(u^{(n)}(r))_n$ which converges weakly in $L^2(\reals \times \Omega)$ for all $r \in [t, T]$. This limiting function would then solve the SPDE \cref{eqn:spdewellposed}. 

Unfortunately the sequence $(u^{(n)})_n$ defined in \cref{eqn:usimple} is not guaranteed to be bounded in $L^2(\Omega; L^2(t, T; H^1)) \cap L^2(\Omega; L^\infty(t, T; L^2(\reals)))$ due to the presence of the operator $\AAscr_i^x$ (the reason for this will be clear later). Hence, what we do is perform the following truncation: For $R > 0$ define
\begin{align*}
	V_r^R := V_r \frac{|V_r| \wedge Rr^k}{|V_r|}\rind{\{r > 0\}}+ V_0 \rind{\{r = 0\}}
\end{align*}
where $k > 0$ is a parameter. It is clear that $V_r^R$ converges to $V_r$ as $R \to \infty$ pointwise in $r$. We then modify the operator $\AAscr_i^x$ with a truncated version of it, namely, 
\begin{align*}
	\AAscr^{R,x}_i := \frac{1}{\Delta t} \int_{t_i}^{t_{i+1}} \frac{\partial_y(p(r,V^R_r) \beta(r,V^R_r))}{p(r,V^R_r)} \BBscr_i^x  \dd r.
\end{align*}
We will write $\AAscr_i^{R, \cdot} \equiv \AAscr_i^R$. We also define the following indicator random variable
\begin{align}
	\gamma_R := \rind{\{\sup_{t \leq r \leq T} |V_r| \leq Rt^k \}}, \label{eqn:gamma}
\end{align}
which we note yields $\gamma_R V_r^R  = \gamma_R V_r$. This suggests that we should define a modified sequence $(u^{(R)}_i(x))_{i}$ through the difference scheme:
\begin{align}
\begin{split}
	u^{(R)}_i(x) - u^{(R)}_{i+1}(x) &= \LLscr_i^xu^{(R)}_i(x) \Delta t - \CCscr_i^x u^{(R)}_{i+1}(x) \Delta t - \AAscr^{R,x}_i u^{(R)}_{i+1}(x) \Delta t 
\\&\quad+  \BBscr_i^x u^{(R)}_{i+1}(x) \Delta \mrB_i, \quad i = n-1, \dots, 0, \\
	u_n^{(R)}(x) &= \varphi(x),
	 \label{eqn:diffschemeR}
\end{split}
\end{align}
where we will write $u^{(R)}_i \equiv u^{(R)}_i(\cdot) \in H^1$ considering $\omega \in \Omega$ as fixed. Moreover, define
\begin{align}
	u^{(R, n)}(r, x) := \sum_{i=0}^{n-1} u^{(R)}_i(x) \gind{[t_i^{(n)}, t_{i+1}^{(n)})}{r} + u_n^{(R)}(x)\gind{\{t_n^{(n)} \}}{r} \label{eqn:usimpleR}
\end{align}
which is simple in $r$ on $\PP_n$ for each $n$. Again, we will write $u^{(R, n)}(r) \equiv u^{(R, n)}(r, \cdot) \in H^1$ where we consider $\omega \in \Omega$ as fixed. 

Thus instead of working with $(u^{(n)})_n$ defined in \cref{eqn:usimple}, we will now work with $(u^{(R,n)})_n$ defined in \cref{eqn:usimpleR}. To reiterate, we intend to prove that $(u^{(R,n)})_n$ is bounded in $L^2(\Omega; L^2(t, T; H^1)) \cap L^2(\Omega; L^\infty(t, T; L^2(\reals)))$. Once this is true, then there will exist a subsequence $(u^{(R, n_j)}(r))_j$ and element $u^{(R)}(r)$ such that $u^{(R, n_j)}(r) \to u^{(R)}(r)$ weakly in $L^2(\reals \times \Omega)$ for all $r \in [t, T]$. It is then not hard to show that the weak limit $u^{(R)}$ will solve the SPDE
\begin{align}
\begin{split}
	-\dd u^{(R)} (t, x) &= \left (\LL^x_t - \CC^x_t - \frac{\partial_y (p(t,V^R_t) \beta(t,V^R_t) )}{p(t,V^R_t)} \BB_t^x \right ) u^{(R)}(t,x)\dd t + \BB^x_t u^{(R)}(t,x) \bd \mrB_t, \label{eqn:spdewellposedR} \\
	 u^{(R)}(T,x) &= \varphi(x).
\end{split}
\end{align}	
Finally, by definition of $\gamma_R$ and $u^{(R, n)}$ we will get  $\gamma_R u^{(R, n)} = \gamma_R u^{(n)}$ and $\gamma_R u^{(R)} = \gamma_R u$.

Now we proceed in proving that $(u^{(R, n)})_n$ is bounded in $L^2(\Omega; L^2(t, T; H^1)) \cap L^2(\Omega; L^\infty(t, T; L^2(\reals)))$. First of all, we have 
\begin{align*}
	\| u^{(R, n)} \|^2_{L^2(\Omega; L^2(t, T; H^1))} &= \Ex \left [ \int_t^T \| u^{(R, n)}(r, \cdot) \|^2_{H^1} \dd r \right ] = \sum_{i = 0}^{n-1} \Ex \left [ \int_{t_i}^{t_{i+1}} \| u^{(R)}_i \|^2_{H^1} \dd r \right ]\\
					&= \sum_{i = 0}^{n-1} \Delta t \Ex \left [ \| u^{(R)}_i \|^2_{H^1} \right ]
\end{align*}
and
\begin{align*}
	\| u^{(R, n)} \|^2_{L^2(\Omega; L^\infty(t, T; L^2(\reals)))} &= \Ex \left [ \sup_{i = 0, 1, \dots, n} \| u_i^{(R)} \|^2_{L^2(\reals)} \right ].
\end{align*}
Recall the variational formulation of the SPDE from \Cref{remark:variational}. Now rearrange the difference scheme \cref{eqn:diffschemeR} as 
\begin{align}
	u^{(R)}_i - u^{(R)}_{i+1} - \left ( \LLscr_i u^{(R)}_i - \CCscr_i u^{(R)}_{i+1} - \AAscr^{R}_i u^{(R)}_{i+1} \right ) \Delta t &=   \BBscr_i u^{(R)}_{i+1} \Delta \mrB_i, \label{eqn:diffschemeRalt}
\end{align}
and then take the square of both sides, yielding the inequality
\begin{align}
\begin{split}
	&\| u^{(R)}_i - u^{(R)}_{i+1} \|_{L^2(\reals)}^2 \\
	&- 2 \Delta t \left ( \left \langle \LLscr_i u^{(R)}_i, u^{(R)}_i - u^{(R)}_{i+1} \right \rangle - \left \langle \CCscr_i u^{(R)}_{i+1}, u^{(R)}_i - u^{(R)}_{i+1} \right \rangle_{L^2(\reals)} - \left \langle \AAscr^{R}_i u^{(R)}_{i+1}, u^{(R)}_i - u^{(R)}_{i+1} \right \rangle_{L^2(\reals)} \right ) \\ &\leq   \| \BBscr_i u^{(R)}_{i+1} \|^2_{L^2(\reals)} (\Delta \mrB_i)^2. \label{eqn:squareddiffschemeR}
\end{split}
\end{align}
Moreover, multiplying \cref{eqn:diffschemeRalt} with $2 u_{i+1}^{(R)}$ yields
\begin{align}
\begin{split}
	&2 \left \langle u_{i+1}^{(R)}, u^{(R)}_i - u^{(R)}_{i+1} \right \rangle_{L^2(\reals)} - 2 \Delta t \left ( \left \langle \LLscr_i u^{(R)}_i, u_{i+1}^{(R)} \right \rangle - \left \langle \CCscr_i u^{(R)}_{i+1}, u_{i+1}^{(R)} \right \rangle_{L^2(\reals)} - \left \langle \AAscr^{R}_i u^{(R)}_{i+1}, u_{i+1}^{(R)} \right \rangle_{L^2(\reals)} \right ) \\
	&=   \langle u_{i+1}^{(R)}, \BBscr_i u^{(R)}_{i+1} \rangle_{L^2(\reals)} \Delta \mrB_i. \label{eqn:innerdiffschemeR}
\end{split}
\end{align}
Adding \cref{eqn:squareddiffschemeR} and \cref{eqn:innerdiffschemeR} together yields
\begin{align*}
	&\| u_i^{(R)} \|^2_{L^2(\reals)} - \| u_{i+1}^{(R)} \|^2_{L^2(\reals)} - 2\Delta t \left ( \left \langle \LLscr_i u^{(R)}_i, u_i^{(R)} \right \rangle - \left \langle \CCscr_i u^{(R)}_{i+1}, u_i^{(R)} \right \rangle_{L^2(\reals)}  - \left \langle \AAscr^{R}_i u^{(R)}_{i+1}, u_i^{(R)} \right \rangle_{L^2(\reals)} \right )\\
	 &\leq \| \BBscr_i u^{(R)}_{i+1} \|^2_{L^2(\reals)} (\Delta \mrB_i)^2 +   \langle u_{i+1}^{(R)}, \BBscr_i u^{(R)}_{i+1} \rangle_{L^2(\reals)} \Delta \mrB_i.
\end{align*}
Now taking expectation and sum of the preceding expression yields
\begin{align}
\label{eqn:exdiffscheme}
\begin{split}
	&\Ex \| u_m^{(R)} \|^2_{L^2(\reals)} - \Ex \| u_n^{(R)} \|^2_{L^2(\reals)} 
	- 2\Delta t \sum_{i = m}^{n-1} \Ex \left ( \left \langle \LLscr_i u^{(R)}_i, u_i^{(R)} \right \rangle - \left \langle \CCscr_i u^{(R)}_{i+1}, u_i^{(R)} \right \rangle_{L^2(\reals)}  - \left \langle \AAscr^{R}_i u^{(R)}_{i+1}, u_i^{(R)} \right \rangle_{L^2(\reals)} \right ) \\
	 &\leq \Ex \sum_{i=m}^{n-1} \| \BBscr_i u^{(R)}_{i+1} \|^2_{L^2(\reals)} \Delta t.
\end{split}
\end{align}
Note that to obtain the right hand side of \cref{eqn:exdiffscheme} we have towered with $\bar \FF_{t_{i+1}, T}^{V, B}$ and used that $\BBscr_i u_{i+1}^{(R)}$ is a $\bar \FF_{t_{i+1}, T}^{V, B}$-measurable random element and that $\Delta \mrB_i$ is independent of $\bar \FF_{t_{i+1}, T}^{V, B}$.

As alluded to before, there are some intricacies with the term $\Ex[ | \langle \AAscr^{R}_i u^{(R)}_{i+1}, u_i^{(R)} \rangle_{L^2(\reals)} |]$. Thankfully our truncation method prevents any difficulties from arising, as
\begin{align*}
	 \| \AAscr_i^{R} u_{i+1}^{(R)} \|^2_{L^2(\reals)} &= \left \| \frac{1}{\Delta t} \int_{t_i}^{t_{i+1}} \frac{\partial_y(p(r,V^R_r) \beta(r,V^R_r))}{p(r,V^R_r)} \BBscr_i u^{(R)}_{i+1} \dd r \right \|^2_{L^2(\reals)} \\
		&= \frac{1}{(\Delta t)^2} \left (\int_{t_i}^{t_{i+1}} \frac{\partial_y(p(r,V^R_r) \beta(r,V^R_r))}{p(r,V^R_r)} \dd r \right )^2  \left \| \BBscr_i u^{(R)}_{i+1}\right \|^2_{L^2(\reals)} \\
		&\leq \frac{1}{(\Delta t)^2} 2C^2 \Delta t \left ( \int_{t_i}^{t_{i+1}} \left (\frac{R^{2p_1}}{r^{2(q_1 -k p_1)}} + \frac{R^{2p_2}}{r^{2(q_2 - kp_2)}} \right) \dd r \right ) \left \| \BBscr_i u^{(R)}_{i+1} \right \|^2_{L^2(\reals)}\\
		&\leq \frac{2C^2}{\Delta t} \left (\int_{t_i}^{t_{i+1}} \left (\frac{R^{2p_1}}{r^{2(q_1 - k p_1)}} + \frac{R^{2p_2}}{r^{2(q_2 - k p_2)}} \right) \dd r \right ) \|u_{i+1}^{(R)}\|^2_{H^1}.
\end{align*}
Note we have used \Cref{ass:ratio_density} in order to obtain the first inequality above, since
\begin{align*}
	\left |\frac{\partial_y(p(r, V_r^R) \beta(r, V_r^R))}{p(r, V_r^R)} \right | \leq C \left ( \frac{|V_r^R|^{p_1}}{r^{q_1}} + \frac{|V_r^R|^{p_2}}{r^{q_2}} \right )  \leq C \left ( \frac{R^{p_1}}{r^{q_1 - kp_1}} + \frac{R^{p_2}}{r^{q_2 - kp_2}} \right ).
\end{align*}
Furthermore by choosing $k >0$ large enough, we have that
\begin{align}
\label{eqn:truncationblowup}
	\int_{t_i}^{t_{i+1}} \left (\frac{R^{2p_1}}{r^{2(q_1 - kp_1)}} + \frac{R^{2 p_2}}{r^{2(q_2 - kp_2)}} \right) \dd r =  \bigO{\Delta t}
\end{align}
due to the conditions imposed on $p_i$ and $q_i$ in \Cref{ass:ratio_density}. Now define 
\begin{equation}
\begin{aligned}[c]
\label{eqn:piecewiseoperators}
	\bar \LL^x(r) &:= \sum_{i = 0}^{n-1} \LLscr_i^x \gind{[t_i, t_{i+1})}{r}, 		\quad & 	\bar \BB^x(r) &:= \sum_{i = 0}^{n-1} \BBscr_i^x \gind{[t_i, t_{i+1})}{r}, \\
	\bar \CC^x(r) &:= \sum_{i = 0}^{n-1} \CCscr_i^x \gind{[t_i, t_{i+1})}{r}, 	\quad &	\bar \AA^{R, x}(r) &:= \sum_{i = 0}^{n-1} \AAscr_i^{R,x} \gind{[t_i, t_{i+1})}{r}.
\end{aligned}
\end{equation}
It is then clear that $\bar \LL : [0, T] \to \mathbb{B}(H^1, H^{-1})$, and $\bar \BB, \bar \CC, \bar \AA^R: [0, T] \to \mathbb{B}(H^1, L^2(\reals))$. At this point the proof follows in a similar manner to the end of \citep[][Lemma 3.1, part II]{pardoux1979stochastic}, which itself is an adaptation of classical existence and uniqueness arguments for parabolic PDEs, a good reference for such arguments can be found in \citep[][\S 7.1]{evans2010partial}. More specifically, the end of the proof involves rewriting \cref{eqn:exdiffscheme} in terms of the operators from \cref{eqn:piecewiseoperators} and appealing to classical energy estimates.

\qed

%-------Remark: Truncation is vital
\begin{remark}
Notice that utilising the truncation $V_r^R$ is vital. Without it, we would not be able to ensure that \cref{eqn:truncationblowup} holds for every $t_i, t_{i+1} \in [t, T]$. As a simple example, consider the case of $V_r = B_r$ without truncation. Then 
\begin{align*}
	\int_{t_i}^{t_{i+1}} \Ex \left [\left ( \frac{\partial_y(p(r, B_r) \beta(r, B_r))}{p(r, B_r)} \right )^2 \right ]\dd r= \int_{t_i}^{t_{i+1}} \Ex \left [ \left (\frac{B_r}{r} \right )^2 \right ] \dd r = \int_{t_i}^{t_{i+1}} \frac{1}{r} \dd r
\end{align*}
which is not $\bigO{\Delta t}$ when $t_i = t$ and $t = 0$.
\end{remark}

%----------------------------------------------------------Proof of Proposition: Conditional Feynman Kac Prop------------------------------------------------------
\subsection*{Proof of \Cref{prop:conditionalfeynmanwellposedprop}}By \Cref{thm:spdeexistence}, there exists a unique $(\bar \FF_{t,T}^{V,B})_{\tinT}$-adapted solution to the SPDE \cref{eqn:spdewellposed} belonging to $L^2(\vep, T ; H^1(\reals)) \cap C([\vep, T]; L^2(\reals))$ for all $\vep > 0$, $\Qro$ a.s., which we will denote by $u(t,x)$. 

Recall the difference scheme \cref{eqn:diffscheme} and $\gamma_R$ defined in \cref{eqn:gamma}. It can be shown that under the additional assumptions \cref{ass:extrareg1} and \cref{ass:extrareg2}, the sequence $\gamma_R u^{(n)}$ is bounded in $L^2 (\Omega ; L^{\infty}(0, T ; H^k(\reals)))$ for all $k \in \naturals \cup \{0 \}$, see \citep[][Lemma 6.3]{pardoux1982equations}. Hence for any $l \in \naturals$, the order of Sobolev space $k$ can be chosen arbitrarily large such that $k > \frac{1}{2} + l$ holds. This implies that the sequence $\gamma_R u^{(n)}$ is in fact bounded in $L^2 (\Omega ; L^{\infty}(0, T ; C^l_b(\reals)))$ via a standard Sobolev embedding theorem.

We will write $\Ex_{t,x}^{t,T}[\cdot ] \equiv \Ex[\cdot | X_t = x, \bar \FF_{t,T}^{V,B}]$. Now consider
\begin{align}
	\gamma_R \Ex_{t,x}^{t,T} \left [ \sum_{i =0}^{n-1} u^{(n)}(t_{i+1}, X_{t_{i+1}}) - u^{(n)}(t_i, X_{t_i}) \right ] = \gamma_R \left ( \Ex_{t,x}^{t,T} [ \varphi(X_T)] - u^{(n)}(t, x) \right ). \label{eqn:exofincrements}
\end{align}
Similar to arguments made in the proof of \Cref{thm:spdeexistence}, as $n \to \infty$ the RHS of \cref{eqn:exofincrements} tends to $ \gamma_R \left ( \Ex_{t,x}^{t,T} [ \varphi(X_T)] - u(t, x) \right )$ weakly in $L^2(\reals \times \Omega)$, pointwise in $t$ along a subsequence, which we will from now on identify with the original sequence. Our task now is to show that the LHS of \cref{eqn:exofincrements} tends to 0 in $L^1(\Qro_{t,x})$ as $n \to \infty$, or equivalently, as $\Delta t \to 0$. We will eventually see that this suffices for proving the proposition.

Focusing on the increment of $u^{(n)}(r, X_r)$ over $[t_i, t_{i+1})$, we can decompose it as follows:
\begin{align*}
	u^{(n)}(t_{i+1}, X_{t_{i+1}}) - u^{(n)}(t_i, X_{t_i}) &= \left [ u^{(n)}(t_{i+1}, X_{t_{i+1}}) - u^{(n)}(t_{i+1}, X_{t_i}) \right ] + \left [u^{(n)}(t_{i+1}, X_{t_i}) - u^{(n)}(t_i, X_{t_i}) \right ] \\
&= \chi_i + \tau_i,
\end{align*}
where 
\begin{align*}
	\chi_i &:= u^{(n)}(t_{i+1}, X_{t_{i+1}}) - u^{(n)}(t_{i+1}, X_{t_i}), & \tau_i &:=u^{(n)}(t_{i+1}, X_{t_i}) - u^{(n)}(t_i, X_{t_i}).
\end{align*}
Notice that for $\chi_i$, space is moving and time is fixed, whereas for $\tau_i$ space is fixed and time is moving. 
%----------chi_i
We can rewrite $\chi_i$ using Taylor's theorem with Lagrange remainder: 
\begin{align*}
	\chi_i =  u^{(n)}(t_{i+1}, X_{t_{i+1}}) - u^{(n)}(t_{i+1}, X_{t_i}) = u^{(n)}_x(t_{i+1}, X_{t_i}) \Delta X_i + \frac{1}{2} u^{(n)}_{xx}(t_{i+1}, H_{t_i}) (\Delta X_i)^2
\end{align*}
where $H_{t_i} \in [X_{t_i}, X_{t_{i+1}}]$.
%----------tau_i
For $\tau_i$, we can use the difference scheme \cref{eqn:diffscheme}, as $\tau_i = u_{i+1}(X_{t_i}) - u_{i}(X_{t_i})$, yielding
\begin{align*}
	 \tau_i =  u^{(n)}(t_{i+1}, X_{t_i}) - u^{(n)}(t_i, X_{t_i}) &= -\LLscr_i^{X_{t_i}} u^{(n)} (t_i, X_{t_i}) \Delta t + \CCscr_i^{X_{t_i}} u^{(n)}(t_{i+1}, X_{t_i}) \Delta t \\&\quad + \AAscr^{X_{t_i}}_i u^{(n)}(t_{i+1}, X_{t_i}) \Delta t -  \BBscr_i^{X_{t_i}} u^{(n)}(t_{i+1}, X_{t_i}) \Delta \mrB_i.
\end{align*}
But $\Delta \mrB_i = \Delta B_i +  \int_{t_i}^{t_{i+1}} \frac{\partial_y(p(r,V_r) \beta(r,V_r))}{p(r,V_r)} \dd r$, which allows us to eliminate the preceding $\AAscr_i^{X_{t_i}}$ term, thus 
\begin{align}
\label{eqn:tauidiff}
\begin{split}
	 \tau_i &=  u^{(n)}(t_{i+1}, X_{t_i}) - u^{(n)}(t_i, X_{t_i}) \\ 
	 &= -\LLscr_i^{X_{t_i}} u^{(n)} (t_i, X_{t_i}) \Delta t + \CCscr_i^{X_{t_i}} u^{(n)}(t_{i+1}, X_{t_i}) \Delta t - \BBscr_i^{X_{t_i}} u^{(n)}(t_{i+1}, X_{t_i}) \Delta B_i.
\end{split}
\end{align}
Now we expand the terms in $\chi_i$ and $\tau_i$. To expand $\chi_i$ we substitute in
\begin{align*}
	\Delta X_i &= \int_{t_i}^{t_{i+1}} \dd X_r 
	\\& = \int_{t_i}^{t_{i+1}} \mu(r, X_r, V_r) \dd r + \int_{t_i}^{t_{i+1}} \rho_r \sigma(r, X_r, V_r) \dd B_r + \int_{t_i}^{t_{i+1}} \varrho_r \sigma(r, X_r, V_r) \dd \hat B_r.
\end{align*}
Furthermore, to expand $\tau_i$ we substitute in the explicit expressions for $\LLscr_i^{X_{t_i}} u^{(n)} (t_i, X_{t_i}) \Delta t$, $\BBscr_i^{X_{t_i}} u^{(n)} (t_{i+1}, X_{t_i}) \Delta B_i$, and $\CCscr_i^{X_{t_i}} u^{(n)} (t_{i+1}, X_{t_i}) \Delta t $, which are
\begin{align*}
	\LLscr_i^{X_{t_i}} u^{(n)} (t_i, X_{t_i}) \Delta t &= \frac{1}{2} u^{(n)}_{xx}(t_i, X_{t_i}) \int_{t_i}^{t_{i+1}} \sigma^2(r, X_{t_i}, V_{t_i}) \dd r + u_x^{(n)}(t_i, X_{t_i}) \int_{t_i}^{t_{i+1}} \mu(r, X_{t_i}, V_{t_i}) \dd r, \\
	\BBscr_i^{X_{t_i}} u^{(n)} (t_{i+1}, X_{t_i}) \Delta B_i &=  u^{(n)}_{x}(t_{i+1}, X_{t_i}) \int_{t_i}^{t_{i+1}} \rho_r \sigma(r, X_{t_i}, V_{t_{i+1}}) \dd r \frac{\Delta B_i}{\Delta t}, \\
	\CCscr_i^{X_{t_i}} u^{(n)} (t_{i+1}, X_{t_i}) \Delta t &= u^{(n)}_{x}(t_{i+1}, X_{t_i}) \int_{t_i}^{t_{i+1}} \rho_r \beta(r, V_{t_i}) \sigma_y(r, X_{t_i}, V_{t_i}) \dd r.
\end{align*}
%----------XX_i, YY_i, ZZ_i, WW_i
Combining $\chi_i$ and $\tau_i$ after the appropriate substitutions finally yields
\begin{align*}
	u^{(n)}(t_{i+1}, X_{t_{i+1}}) - u^{(n)}(t_i, X_{t_i}) = \XX_i^{(n)} + \YY_i^{(n)} + \ZZ_i^{(n)} + \WW_i^{(n)},
\end{align*}
where
\begin{align*}
	\XX_i^{(n)} &:=  u_x^{(n)}(t_{i+1}, X_{t_i}) \int_{t_i}^{t_{i+1}} \mu(r, X_r, V_r) \dd r - u_x^{(n)}(t_i, X_{t_i}) \int_{t_i}^{t_{i+1}} \mu(r, X_{t_i}, V_{t_i}) \dd r,  \\
	\YY_i^{(n)} &:= \frac{1}{2} u_{xx}^{(n)}(t_{i+1}, H_{t_i}) (\Delta X_i)^2 - \frac{1}{2} u_{xx}^{(n)}(t_i, X_{t_i}) \int_{t_i}^{t_{i+1}} \sigma^2(r, X_{t_i}, V_{t_i}) \dd r,  \\
%		\YY_i^{(n)} &:= \frac{1}{2} u_{xx}^{(n)}(t_{i+1}, X_{t_i}) \int_{t_i}^{t_{i+1}} \sigma^2(r, X_r, V_r) \dd r - \frac{1}{2} u_{xx}^{(n)}(t_i, X_{t_i}) \int_{t_i}^{t_{i+1}} \sigma^2(r, X_{t_i}, V_{t_i}) \dd r,  \\
	\ZZ_i^{(n)} &:= u_x^{(n)}(t_{i+1}, X_{t_i}) \int_{t_i}^{t_{i+1}} \rho_r \sigma(r, X_r, V_r) \dd B_r - u_x^{(n)}(t_{i+1}, X_{t_i}) \int_{t_i}^{t_{i+1}} \rho_r \sigma(r, X_{t_i}, V_{t_{i+1}}) \dd r \frac{\Delta B_i}{\Delta t} \\& \quad + u_x^{(n)}(t_{i+1}, X_{t_i})  \int_{t_i}^{t_{i+1}} \rho_r \beta(r, V_{t_i}) \sigma_y(r, X_{t_i}, V_{t_i}) \dd r, \\
	\WW_i^{(n)} &:=  u_x^{(n)}(t_{i+1}, X_{t_i}) \int_{t_i}^{t_{i+1}} \varrho_r \sigma(r, X_r, V_r) \dd \hat B_r.
\end{align*}
Thus \cref{eqn:exofincrements} can be rewritten as
\begin{align}
	\gamma_R \Ex_{t,x}^{t,T} \left [ \sum_{i =0}^{n-1}   \XX_i^{(n)} + \YY_i^{(n)} + \ZZ_i^{(n)} + \WW_i^{(n)}   \right ] = \gamma_R \left ( \Ex_{t,x}^{t,T} [ \varphi(X_T)] - u^{(n)}(t, x) \right ). \label{eqn:exofincrements2}
\end{align}
Note that as $\gamma_R \leq 1$ it suffices to show that 
\begin{align*}
 	 \Ex_{t,x}^{t,T} \left [ \sum_{i =0}^{n-1}   \XX_i^{(n)} \right ], \Ex_{t,x}^{t,T} \left [ \sum_{i =0}^{n-1}   \YY_i^{(n)} \right ], \Ex_{t,x}^{t,T} \left [ \sum_{i =0}^{n-1} \ZZ_i^{(n)} \right ], \Ex_{t,x}^{t,T} \left [ \sum_{i =0}^{n-1} \WW_i^{(n)} \right ]
\end{align*}
each converge to $0$ in $L^1(\Qro_{t,x})$ as $\Delta t \to 0$, which we will do case by case. Note that we can immediately ignore $\WW_i^{(n)}$ as it will be zero after taking $\Ex_{t,x}^{t,T}$ and then towering with $\Ex_{t,x}^{t,T}[\cdot | X_{t_i}]$, due to the independence of $\bar \FF_{t,T}^{V,B}$ and $\hat B$.

It should be clear as to why we reexpressed \cref{eqn:exofincrements} as \cref{eqn:exofincrements2}. From the forms of $\XX_i^{(n)}$ and $\YY_i^{(n)}$, one can already postulate that 
\begin{align*}
	\Ex_{t,x}^{t,T} \sum_{i=0}^{n-1} \XX_i^{(n)} &\longrightarrow 0 \quad \text{ and } \quad  \Ex_{t,x}^{t,T} \sum_{i=0}^{n-1} \YY_i^{(n)} \longrightarrow 0
\end{align*}
in $L^1(\Qro_{t,x})$. The term $\ZZ_i^{(n)}$ is more puzzling; essentially there is an extra term from the SPDE \cref{eqn:spdewellposed} given through $\CC_t^x$ due to time reversal of the stochastic integral w.r.t. $B$, this extra term essentially being the quadratic covariation of $B$ and the corresponding integrand.

Note through the tower property we have 
\begin{align*}
	\Ex_{t,x} \left | \sum_{i=0}^{n-1} \Ex_{t,x}^{t,T} \left [ \cdot \right ] \right | \leq \sum_{i=0}^{n-1} \Ex_{t,x} | \cdot |.
\end{align*}
Hence, in order to prove the proposition, it is sufficient to show that terms within the summation are $\littleo{\Delta t}$. Furthermore, it will often suffice to neglect second-order terms when applying It\^o's formula and simply write them as $\bigO{\Delta t}$, since applying a Riemann or It\^o integration to a $\bigO{\Delta t}$ term over $[t_i, t_{i+1}]$ yields a $\littleo{\Delta t}$ term. Moreover, to get some intuition as to whether terms will contribute or not, one should preemptively attempt to determine each integral's order of contribution, noting that the iteration of integrals (whether it be Riemann or It\^o) will decrease that term's order of contribution.

Before proceeding, recall that the sequence $\gamma_R u^{(n)}$ is bounded in $L^2 (\Omega ; L^{\infty}(0, T ; C^l_b(\reals)))$ for any $l \in \naturals$. This ensures that any terms we encounter involving $u^{(n)}$ and its partial derivatives w.r.t. $x$ in the summation do not explode as $\Delta t \to 0$ in $L^1(\Qro_{t,x})$, noting that we can bring in $\gamma_R$ into our calculations if necessary by \cref{eqn:exofincrements2}.

%----------XX_i to zero
$\recbullet$ We will first show $\Ex_{t,x}^{t,T} \sum_{i=0}^{n-1} \XX_i^{(n)}$ tends to 0 in $L^1(\Qro_{t,x})$. By It\^o's formula, we can rewrite 
\begin{align*}
	\mu(r, X_r, V_r) &= \mu(r, X_{t_i}, V_{t_i}) + \int_{t_i}^r \mu_x(r, X_\theta, V_\theta) \dd X_\theta + \int_{t_i}^r \mu_y(r, X_\theta, V_\theta) \dd V_\theta + \bigO{\Delta t}.
\end{align*}
Substituting this into the expression for $\XX_i^{(n)}$ yields
\begin{align}
\begin{split}
	\XX_i^{(n)} &= \big (u_x^{(n)}(t_{i+1}, X_{t_i})  - u_x^{(n)}(t_{i}, X_{t_i}) \big )  \int_{t_i}^{t_{i+1}} \mu(r, X_{t_i}, V_{t_i}) \dd r \\
	&\quad + u_x^{(n)}(t_{i+1}, X_{t_i})   \int_{t_i}^{t_{i+1}} \left (\int_{t_i}^r \mu_x(r, X_\theta, V_\theta) \dd X_\theta + \int_{t_i}^r \mu_y(r, X_\theta, V_\theta) \dd V_\theta \right ) \dd r + \littleo{\Delta t}. \label{eqn:XXreexp}
\end{split}
\end{align}
Note the $\bigO{\Delta t}$ term has become $\littleo{\Delta t}$ after applying $\int_{t_i}^{t_{i+1}} (\cdots) \dd r$ to it. 

We now focus on the first term on the RHS of \cref{eqn:XXreexp}. In order to treat it, we first recognise that $\mu$ is bounded. Ergo, it is now enough to show that $ u_x^{(n)}(t_{i+1}, X_{t_i})  - u_x^{(n)}(t_{i}, X_{t_i}) = \littleo{1}.$ This follows from noting that $-(u_{i}(X_{t_i}) - u_{i+1}(X_{t_i}))=u^{(n)}(t_{i+1}, X_{t_i}) - u^{(n)}(t_i, X_{t_i})$, and differentiating \cref{eqn:diffscheme} in $x$. Since $\gamma_R u^{(n)}$ is bounded in $L^2(\Omega ; L^{\infty}(0, T;C_b^3(\reals)))$, we can conclude that the term is at least $\littleo{1}$. This yields
\begin{align*}
	\big (u_x^{(n)}(t_{i+1}, X_{t_i})  - u_x^{(n)}(t_{i}, X_{t_i}) \big )  \int_{t_i}^{t_{i+1}} \mu(r, X_{t_i}, V_{t_i}) \dd r & \leq C \Delta t \big (u_x^{(n)}(t_{i+1}, X_{t_i})  - u_x^{(n)}(t_{i}, X_{t_i}) \big )  \\&= \littleo{\Delta t}.
\end{align*}
For the next term in \cref{eqn:XXreexp} we can expand this out to get
\begin{align}
	 u_x^{(n)}(t_{i+1}, X_{t_i})   \int_{t_i}^{t_{i+1}} \left (\int_{t_i}^r \mu_x(r, X_\theta, V_\theta) \dd X_\theta + \int_{t_i}^r \mu_y(r, X_\theta, V_\theta) \dd V_\theta \right ) \dd r \nonumber\\
	 = u_x^{(n)}(t_{i+1}, X_{t_i}) \int_{t_i}^{t_{i+1}} \left (\int_{t_i}^r a_{r,\theta} \dd \theta + \int_{t_i}^r b_{r, \theta} \dd B_\theta + \int_{t_i}^r c_{r, \theta} \dd \hat B_\theta \right ) \dd r \label{eqn:XXalt1}
 \end{align}
where for example 
\begin{align*}
	a_{r,\theta} = \mu_x(r, X_\theta, V_\theta) \mu(\theta, X_\theta, V_\theta) + \mu_y(r, X_\theta, V_\theta) \alpha(\theta, V_\theta)
\end{align*}
and we can obtain $b_{r,\theta}$ and $c_{r, \theta}$ in a similar fashion. However, their explicit expressions are not important, we just need that they are bounded, and thus we omit writing them. It is simple to show that the $\dd \hat B$ integral term in \cref{eqn:XXalt1} is zero after taking $\Ex_{t,x}^{t,T}$ and then towering with $\Ex_{t,x}^{t,T}[\cdot | X_{t_i}]$. Focusing on the $\dd B$ integral term in \cref{eqn:XXalt1} we have
\begin{align*}
	&\Ex_{t,x}\left | \sum_{i=0}^{n-1} \Ex_{t,x}^{t,T} \left [u_x^{(n)}(t_{i+1}, X_{t_i}) \int_{t_i}^{t_{i+1}} \left (\int_{t_i}^r b_{r,\theta} \dd B_\theta \right ) \dd r   \right ] \right | \\
	&\leq \sum_{i=0}^{n-1} \Ex_{t,x} \left | u_x^{(n)}(t_{i+1}, X_{t_i}) \int_{t_i}^{t_{i+1}} \left (\int_{t_i}^r b_{r,\theta} \dd B_\theta \right ) \dd r  \right | \\
	& \leq \sum_{i=0}^{n-1} \left ( \Ex_{t,x} \left [ u_x^{(n)}(t_{i+1}, X_{t_i}) \right ]^2 \right )^{1/2}  
	\left ( \Ex_{t, x} \left [ \int_{t_i}^{t_{i+1}} \left (\int_{t_i}^r b_{r,\theta} \dd B_\theta \right ) \dd r  \right ]^2 \right )^{1/2}.
\end{align*}
Using Jensen's inequality we have
\begin{align*}
	\Ex_{t, x} \left ( \int_{t_i}^{t_{i+1}} \left (\int_{t_i}^r b_{r,\theta} \dd B_\theta \right ) \dd r  \right )^2 \leq \Delta t \int_{t_i}^{t_{i+1}} \Ex_{t,x} \left ( \int_{t_i}^r b_{r, \theta} \dd B_\theta \right )^2 \dd r = \Delta t \int_{t_i}^{t_{i+1}} \left ( \int_{t_i}^r \Ex_{t,x} (b^2_{r, \theta}) \dd \theta \right )  \dd r.
\end{align*}
Thus we have 
\begin{align*}
	 u_x^{(n)}(t_{i+1}, X_{t_i}) \int_{t_i}^{t_{i+1}} \left (\int_{t_i}^r b_{r,\theta} \dd B_\theta \right ) \dd r  = \littleo{\Delta t}.
\end{align*}
A similar method yields that the expression involving the $\dd \theta$ integral term in \cref{eqn:XXalt1} is $\littleo{\Delta t}$. 

%----------YY_i to zero
$\recbullet$ Showing $\Ex_{t,x}^{t,T} \sum_{i=0}^{n-1} \YY_i^{(n)}$ converges to $0$ in $L^1(\Qro_{t,x})$ as $\Delta t \to 0$ follows in a similar manner to the case pertaining to $\XX_i^{(n)}$, thus we omit it.
%----------ZZ_i to zero

$\recbullet$ Lastly, we show that $\Ex_{t,x}^{t,T} \sum_{i=0}^{n-1} \ZZ_i^{(n)} \to 0$ in $L^1(\Qro_{t,x})$. Focusing on the second term in $\ZZ_i^{(n)}$, note that we can rewrite
\begin{align*}
	\sigma(r, X_{t_i}, V_{t_{i+1}})  &= \sigma(r, X_{t_i}, V_{t_i}) + \int_{t_i}^{t_{i+1}} \sigma_y(r, X_{t_i}, V_{\theta}) \dd V_\theta + \bigO{\Delta t} \\
	&=  \sigma(r, X_{t_i}, V_{t_i}) + \int_{t_i}^{t_{i+1}} \beta(\theta, V_\theta) \sigma_y(r, X_{t_i}, V_\theta) \dd B_\theta + \bigO{\Delta t}.
\end{align*}
Thus the second term in $\ZZ_i^{(n)}$ can be reexpressed as 
\begin{align*}
	&u_x^{(n)}(t_{i+1}, X_{t_i}) \int_{t_i}^{t_{i+1}} \rho_r \sigma(r, X_{t_i}, V_{t_{i+1}}) \dd r \frac{\Delta B_i}{\Delta t} 
\\&= u_x^{(n)}(t_{i+1}, X_{t_i}) \left [\int_{t_i}^{t_{i+1}} \rho_r \sigma(r, X_{t_i}, V_{t_i}) \dd r + \int_{t_i}^{t_{i+1}} \rho_r \left (\int_{t_i}^{t_{i+1}}\beta(\theta, V_\theta) \sigma_y(r, X_{t_i}, V_\theta) \dd B_\theta \right ) \dd r \right ] \frac{\Delta B_i}{\Delta t} \\
&\quad+ o(\Delta t).
\end{align*}
Hence we can reexpress $\ZZ_i^{(n)}$ as
\begin{align}
	\ZZ_i^{(n)} =  \hat \ZZ_i^{(n)} + \bar \ZZ_i^{(n)} + \littleo{\Delta t}, \label{eqn:ZZalt1}
\end{align} 
where
\begin{align*}
	\hat \ZZ_i^{(n)} &:=  u_x^{(n)}(t_{i+1}, X_{t_i}) \left [ \int_{t_i}^{t_{i+1}} \rho_r \sigma(r, X_r, V_r) \dd B_r - \int_{t_i}^{t_{i+1}} \rho_r \sigma(r, X_{t_i}, V_{t_i}) \dd r   \frac{\Delta B_i}{\Delta t} \right ], \\
\bar \ZZ_i^{(n)} &:=  u_x^{(n)}(t_{i+1}, X_{t_i}) \Bigg [ \int_{t_i}^{t_{i+1}} \rho_r \beta(r, V_{t_i}) \sigma_y(r, X_{t_i}, V_{t_i}) \dd r \\&\quad - \int_{t_i}^{t_{i+1}} \rho_r \left (\int_{t_i}^{t_{i+1}}\beta(\theta, V_\theta) \sigma_y(r, X_{t_i}, V_\theta) \dd B_\theta \right ) \dd r \frac{\Delta B_i}{\Delta t} \Bigg].
\end{align*}
We can rewrite $\hat \ZZ_i^{(n)}$ and $\hat \ZZ_i^{(n)}$ by pulling the integrals out to the front:
\begin{align*}
	\hat \ZZ_i^{(n)} &= u_x^{(n)}(t_{i+1}, X_{t_i}) \int_{t_i}^{t_{i+1}} \frac{1}{\Delta t} \left ( \int_{t_i}^{t_{i+1}} \rho_r \sigma(r, X_r, V_r) - \rho_\theta \sigma(\theta , X_{t_i}, V_{t_i}) \dd \theta \right )\dd B_r, \\
	\bar \ZZ_i^{(n)} &=  u_x^{(n)}(t_{i+1}, X_{t_i}) \int_{t_i}^{t_{i+1}} \rho_r \left [ \int_{t_i}^{t_{i+1}}\left ( \frac{1}{\Delta B_i} \beta(r, V_{t_i}) \sigma_y(r, X_{t_i}, V_{t_i}) - \frac{\Delta B_i}{\Delta t} \beta(\theta, V_\theta) \sigma_y(r, X_{t_i}, V_\theta)  \right ) \dd B_\theta \right ] \dd r.
\end{align*}

Focusing on $\hat \ZZ_i^{(n)}$, we can rewrite the integrand as:
\begin{align*}
	\rho_r \sigma(r, X_r, V_r) - \rho_\theta \sigma(\theta , X_{t_i}, V_{t_i}) &= \left [ \rho_r \sigma(r, X_r, V_r) - \rho_{t_i} \sigma(t_i , X_{t_i}, V_{t_i}) \right ] - \left [\rho_\theta \sigma(\theta, X_{t_i}, V_{t_i}) - \rho_{t_i} \sigma(t_i , X_{t_i}, V_{t_i}) \right ] \\
	&= \int_{t_i}^r a_\nu \dd B_\nu + \int_{t_i}^r b_\nu \dd \hat B_\nu + \bigO{\Delta t},
\end{align*}
where the $\bigO{\Delta t}$ term contains the second-order terms from applying It\^o's formula on the preceding $r$ term (i.e., first term), as well as the $\theta$ term (i.e., second term). Both $a_\nu$ and $b_\nu$ are bounded, and their explicit forms are not important. Hence,
\begin{align*}
	\hat \ZZ_i^{(n)} &=  u_x^{(n)}(t_{i+1}, X_{t_i}) \int_{t_i}^{t_{i+1}} \frac{1}{\Delta t} \left ( \int_{t_i}^{t_{i+1}} \left [ \int_{t_i}^r a_\nu \dd B_\nu + \int_{t_i}^r b_\nu \dd \hat B_\nu \right ] \dd \theta \right )\dd B_r + \littleo{\Delta t} \\
	&=  u_x^{(n)}(t_{i+1}, X_{t_i}) \int_{t_i}^{t_{i+1}} \left ( \int_{t_i}^r a_\nu \dd B_\nu \right ) \dd B_r +  u_x^{(n)}(t_{i+1}, X_{t_i}) \int_{t_i}^{t_{i+1}} \left (\int_{t_i}^r b_\nu \dd \hat B_\nu \right )\dd B_r + \littleo{\Delta t}.
\end{align*}
The preceding term involving the $\dd \hat B$ It\^o integral will be zero after one applies $\Ex_{t,x}^{t,T}[\cdot]$ to it and then towers with $\Ex_{t,x}^{t,T} [ \cdot | X_{t_i}]$. Note that
\begin{align*}
	 \int_{t_i}^{t_{i+1}} \left ( \int_{t_i}^r a_\nu \dd B_\nu \right ) \dd B_r & =  \int_{t_i}^{t_{i+1}} \left ( \int_{t_i}^r (a_\nu - a_{t_i}) + a_{t_i} \dd B_\nu \right ) \dd B_r  \\
					&=  \int_{t_i}^{t_{i+1}} \left ( \int_{t_i}^r (a_\nu - a_{t_i}) \dd B_\nu \right ) \dd B_r  + \frac{1}{2} a_{t_i} \left (\Delta B_i^2 - \Delta t \right ). \\
\end{align*}
Hence we can bound $\Ex_{t,x}[\cdot]$ of the $a_\nu$ term like:
\begin{align*}
	&\Ex_{t,x} \left | u_x^{(n)}(t_{i+1}, X_{t_i}) \int_{t_i}^{t_{i+1}} \left ( \int_{t_i}^r a_\nu \dd B_\nu \right ) \dd B_r \right |\\
	&= \Ex_{t,x} \left | u_x^{(n)}(t_{i+1}, X_{t_i})  \left (  \int_{t_i}^{t_{i+1}} \left ( \int_{t_i}^r (a_\nu - a_{t_i}) \dd B_\nu \right ) \dd B_r  + \frac{1}{2} a_{t_i} \left (\Delta B_i^2 - \Delta t \right ) \right ) \right |\\
	&\leq \left ( \Ex_{t,x} \left [  u_x^{(n)}(t_{i+1}, X_{t_i})  \right ]^2 \right )^{1/2} \Bigg [ \left (\Ex_{t,x} \left [ \int_{t_i}^{t_{i+1}} \left ( \int_{t_i}^r (a_\nu - a_{t_i}) \dd B_\nu \right ) \dd B_r \right ]^2 \right )^{1/2} \\
	&\qquad+ \frac{1}{2} \left (\Ex_{t,x} \left [ a_{t_i} (\Delta B_i^2 - \Delta t) \right ]^2 \right )^{1/2} \Bigg]\\
	&= \left ( \Ex_{t,x} \left [  u_x^{(n)}(t_{i+1}, X_{t_i})  \right ]^2 \right )^{1/2} \left [ \left (\int_{t_i}^{t_{i+1}} \left ( \int_{t_i}^r \Ex_{t,x} [a_\nu - a_{t_i}]^2 \dd \nu \right ) \dd r \right )^{1/2} + \frac{1}{2} \left (\Ex_{t,x} \left [ a_{t_i} (\Delta B_i^2 - \Delta t) \right ]^2 \right )^{1/2} \right ].
\end{align*}
From the above calculations, and due to the regularity of $a$, it is now clear that
\begin{align*}
	 u_x^{(n)}(t_{i+1}, X_{t_i}) \int_{t_i}^{t_{i+1}} \left ( \int_{t_i}^r (a_\nu - a_{t_i}) \dd B_\nu \right ) \dd B_r  = \littleo{\Delta t}.
\end{align*}
Furthermore, as a consequence of the quadratic variation of Brownian motion,
\begin{align*}
	 u_x^{(n)}(t_{i+1}, X_{t_i})  a_{t_i} (\Delta B_i^2 - \Delta t) = \littleo{\Delta t}.
\end{align*}
The term $\bar \ZZ_i^{(n)}$ can be tackled in a similar manner to $\hat \ZZ_i^{(n)}$, albeit in a more tedious fashion. Thus we omit it.

In total, we have shown that the LHS of \cref{eqn:exofincrements2} converges to $0$ in $L^1(\Qro_{t,x})$ for all $R > 0$. However, we also have that the RHS of \cref{eqn:exofincrements2} converges to $\gamma_R \left ( \Ex_{t,x}^{t,T} [ \varphi(X_T)] - u(t, x) \right )$ weakly in $L^2(\reals \times \Omega)$, for all $R > 0$. Hence we can conclude that $u(t,x) = \Ex[\varphi(X_T) | X_t = x, \bar \FF_{t,T}^{V,B} ]$ for all $t \in (0, T]$ and $x \in \reals$, $\Qro$ a.s.

\qed

%------------------------------------------------------Proof of Proposition: Conditional Feynman Kac Prop----------------------------------------------------------------
\subsection*{Proof of \Cref{thm:conditionalfeynmanwellposed}}
By \Cref{thm:spdeexistence}, there exists a unique $(\bar \FF_{t,T}^{V,B})_{\tinT}$-adapted solution to the SPDE \cref{eqn:spdewellposed} belonging to  $L^2(\vep, T; H^1(\reals)) \cap C([\vep, T]; L^2(\reals))$ for all $\vep > 0$, $\Qro$ a.s., which we will denote by $u(t,x)$. For simplicity, we will assume that $\varphi \in C_c^\infty(\reals)$; the general case would follow from a standard approximation argument. 

The idea is now classical, one considers a sequence of coefficients
\begin{align}
	\mu^{(m)}, \sigma^{(m)}, \alpha^{(m)}, \beta^{(m)}, \rho^{(m)}, \label{eqn:seqcoefs}
\end{align}
that satisfy the additional assumptions \cref{ass:extrareg1} and \cref{ass:extrareg2} from \Cref{prop:conditionalfeynmanwellposedprop}, are bounded uniformly by constants not depending on $m$, and which converge uniformly on compacts to the original coefficients $\mu, \sigma, \alpha, \beta, \rho$ respectively from the system \crefrange{eqn:system2X}{eqn:system2V}, where we reiterate that the latter only satisfy \Crefrange{ass:sde}{ass:ratio_density}. Denote by $\Qro_{t,x}^{(m)} \equiv \Qro^{(m)}(\cdot | X_t = x)$ the solution of the martingale problem associated with the system \crefrange{eqn:system2X}{eqn:system2V} with the new coefficients \cref{eqn:seqcoefs}. Denote the expectation under $\Qro_{t,x}^{(m)}(\cdot | X_t = x)$ by $\Ex^{(m)}_{t,x}$. It is well known that the sequence $\Qro^{(m)}_{t,x}$ converges weakly to $\Qro_{t,x}$, see for example \citep[][Theorem 11.1.4]{stroock1997multidimensional}. Then denote by $u^{(m)}(t, x)$ the solution to the SPDE \cref{eqn:spdewellposed} associated with the new coefficients \cref{eqn:seqcoefs}. By \Cref{prop:conditionalfeynmanwellposedprop} we have
\begin{align*}
	u^{(m)}(t,x) = \Ex^{(m)} \left [ \varphi(X_T) | \bar \FF_{t,T}^{V, B}, X_t = x \right ],
\end{align*}
for all $t \in (0, T]$ and $x \in \reals$, $\Qro^{(m)}$ a.s.

Let $A_R = \{\sup_{t \leq r \leq T} |V_r| \leq Rt^k \}$ so that \cref{eqn:gamma} can be written as $\gamma_R = \rind{A_R}$. Suppose $\xi$ is an arbitrary $\bar \FF_{t,T}^{V,B}$-measurable continuous random variable with $\xi = \xi \gamma_R$. That is, $\xi(A_R^c) = 0 $. In other words, $\xi$ vanishes outside of the event $A_R$. Then as of consequence of the definition of conditional expectation,
\begin{align}
	\Ex_{t,x} [ u^{(m)}(t,x) \xi ] = \Ex^{(m)}_{t,x} \left [ \varphi \left (X_T \right ) \xi \right ] \label{eqn:weaklim}
\end{align}
where we also note that the restriction of $\Qro^{(m)}$ to $\bar \FF_{t, T}^{V, B}$ does not depend on $m$.
Moreover, it is not hard to see that $\gamma_R u^{(m)}(t, \cdot) \to \gamma_R u(t, \cdot)$ weakly in $L^2(\reals \times \Omega)$ for all $t$ and $R > 0$. Since $\xi = \xi \gamma_R$, we can take limit on the LHS of \cref{eqn:weaklim}, as well as utilise the Portmanteau theorem (which is justified due to the regularity of $\varphi$), which yields
\begin{align*}
	\Ex_{t,x} [ u(t,x) \xi ] = \Ex_{t,x} \left [ \varphi \left (X_T \right ) \xi \right ],
\end{align*}
for all $t \in (0, T]$, $\dd x \times \dd \Qro$ a.e. The result then follows by definition of conditional expectation, where we recognise that the $\sigma$-algebra generated by the collection of preimages of $\xi$ for various $R > 0$ generates $\bar \FF_{t,T}^{V, B}$.
\qed

%---------------------------------------------------------------------------Multivariable setting----------------------------------------------------------------
%------------------------------------------------------------------------------------------------------------------------------------------------------------------
%------------------------------------------------------------------------------------------------------------------------------------------------------------------
\section{Multivariable setting}
\label{sec:multivariablesetting}
\noindent

\noindent Our main results from \Cref{sec:mainresults} can be extended to the multivariable setting. Consider the multivariable diffusion $(X, V)$ taking values in $\reals^{N} \times \reals^{D}$ given by the (forward) system
\begin{align}
	\dd X_t &= \mu(t, X_t, V_t) \dd t + \tilde \sigma(t, X_t, V_t) \dd B_t + \hat \sigma(t, X_t, V_t) \dd \hat B_t, \label{eqn:systemmultiX} \\
	\dd V_t &= \alpha(t, V_t) \dd t + \beta(t, V_t) \dd B_t, \label{eqn:systemmultiV}
\end{align}
where $(B, \hat B)$ is a $\reals^{D} \times \reals^{N}$ valued Brownian motion and
\begin{itemize}
\item $\mu: [0, T] \times \reals^N \times \reals^D \to \reals^{N}$, $\tilde \sigma: [0, T] \times \reals^N \times \reals^D \to \reals^{N \times D}$, 
$\hat \sigma: [0, T] \times \reals^N \times \reals^D \to \reals^{N\times N}$ are each Borel measurable,
\item $\alpha: [0, T] \times \reals^D \to \reals^D$, $\beta: [0, T] \times \reals^D \to \reals^{D \times D}$ are each Borel measurable.
\end{itemize}
Moreover, let $a:= \tilde \sigma \tilde \sigma^\top + \hat \sigma \hat \sigma^\top$.
\begin{remark}
We recover the system \crefrange{eqn:system2X}{eqn:system2V} by choosing $N = D = 1$ as well as $\tilde \sigma = \rho \sigma$ and $\hat \sigma = \sqrt{1 - \rho^2} \sigma$ in the system \crefrange{eqn:systemmultiX}{eqn:systemmultiV}.
\end{remark}

Suppose $V_t$ possesses a density $p(t,y)$ w.r.t. Lebesgue measure. That is, $\Qro(V_t \in A) = \int_A p(t,y) \dd y$ for any Borel set $A$ in $\reals^D$. Similar to the univariate case, we define $\bar \FF_{t, T}^{V, B} = \FF_{t, T}^B \vee \sigma(V_t)$ and 
\begin{align*}
	\mrB_t^k = B_t^k - B_T^k - \int_t^T \frac{\sum_{l = 1}^D \partial_{y_l}(p(r, V_r) \beta_{l,k}(r, V_r))}{p(r, V_r)} \dd r, \quad k = 1, \dots, D.
\end{align*}

Consider the following (backward) SPDE:
\begin{align}
\begin{split}
	-\dd u (t, x) &= \left (\LL^x_t - \CC^x_t - \sum_{ k, l = 1}^D \frac{\partial_{y_l} (p(t,V_t) \beta_{l,k}(t,V_t))}{p(t,V_t)} \left (\BB_t^x\right )_k \right ) u(t,x)\dd t + \sum_{k = 1}^D \left (\BB^x_t \right )_k u(t,x) \bd \mrB^k_t, \label{eqn:spdewellposedmulti} \\
	 u(T,x) &= \varphi(x),
\end{split}
\end{align}
where we have the (stochastic) differential operators
\begin{align*}
	\LL^x_t &:= \frac{1}{2} \sum_{i,j = 1}^N a_{i,j}(t,x,V_t) \partial_{x_i x_j}^2 + \sum_{i = 1}^N \mu_i (t, x, V_t) \partial_{x_i}, \\
	\left ( \BB^x_t\right )_k &:=  \sum_{i =1}^N \tilde \sigma_{i, k} (t, x, V_t) \partial_{x_i}, \quad k = 1, \dots , D,\\
	\CC^x_t &:= \sum_{i =1}^N \sum_{p, q =1}^D \beta_{p, q} (t,V_t)  \left (\partial_{y_p} \tilde \sigma_{i, q} (t, x, V_t) \right) \partial_{x_i}.
\end{align*}

The following assumptions are the multivariable counterparts of \Crefrange{ass:sde}{ass:ratio_density}. However, we can no longer appeal to the Yamada-Watanabe condition for $V$ in \cref{eqn:systemmultiV} as we are in a higher dimensional framework. Instead we will resort to the usual It\^o style existence results. Note that below, $| \cdot |$ refers to the Euclidean norm whereas $\| \cdot \|$ refers to the Frobenius norm.%-
\footnote{For a $m \times n$ real valued matrix $A$, the Frobenius norm (or $L^{2, 2}$ norm) is $\| A \| := \left ( \sum_{i = 1}^m \sum_{j = 1}^n A_{i, j}^2 \right )^{1/2}$.} It should be clear that any analytical properties listed below are considered w.r.t. these norms. Typically $x$ and $y$ denote a point in $\reals^N$ and $\reals^D$ respectively, so that $(x, y)$ denotes a point in $\reals^{N+D}$.

%--------Assumption: multi SDE
\begin{assumptionm}
\label{ass:multi_sde}
\noindent
\begin{enumerate}[label = (mA\arabic*), ref = mA\arabic*]
\item $(x, y) \mapsto \mu(t, x, y)$, $(x, y) \mapsto \tilde \sigma(t, x, y)$ and $(x, y) \mapsto \hat \sigma(t, x, y)$ are locally Lipschitz continuous, uniformly in $t$. 
\item $y \mapsto \alpha(t, y)$ and $y \mapsto \beta(t, y)$ are locally Lipschitz continuous, uniformly in $t$. \label{ass:mA2}
\item $|\mu(t, x, y)| + \| \tilde \sigma(t, x, y) \|  + \| \hat \sigma(t, x, y) \| \leq C(1 + |(x, y)|)$, uniformly in $t$.
\item $|\alpha(t, y)| + \| \beta(t, y) \| \leq C(1 + |y|)$, uniformly in $t$. \label{ass:mA4}
\end{enumerate}
\end{assumptionm}

%---------Assumption: multi \mrB welldefined
\begin{assumptionm}
\label{ass:multi_mrB}
\noindent
\begin{enumerate}[label = (mB\arabic*), ref = mB\arabic*]
\item The density of $V_0$,  $p_0(y) \equiv p(0,y)$ satisfies $\int_{\reals^D} \frac{p^2_0(y)}{1 + |y|^k} \dd y < \infty$ for some $k \in \naturals$.
\item  $\partial^2_{y_i y_j}(\beta \beta^\top)_{i,j}\in L^{\infty}([0, T] \times \reals^D ; \reals)$ for $i, j = 1, \dots, D$.
\end{enumerate}
\end{assumptionm}
By \Cref{thm:definingmrB}, $\mrB$ is a backward Brownian motion in $(\bar \FF_{t,T}^{V,B})_{\tinT}$.

%-------Assumption: multi SPDE
\begin{assumptionm}
\label{ass:multi_spde}
\noindent
\begin{enumerate}[label = (mC\arabic*), ref = mC\arabic*]
\item  $\varphi \in C_c^1(\reals^N; \reals)$. 
\item $\mu \in L^\infty([0, T] \times \reals^N \times \reals^D ; \reals^D)$, $\tilde \sigma \in L^\infty([0, T] \times \reals^N \times \reals^D ; \reals^{N \times D})$, $\hat \sigma \in L^\infty([0, T] \times \reals^N \times \reals^D ; \reals^{D \times D})$ and $\alpha \in L^\infty([0, T] \times \reals^D ; \reals^D)$, $\beta \in L^\infty([0, T] \times \reals^D ; \reals^{D \times D})$.
\item $\partial_{x_i} \tilde \sigma_{i, j} \in L^{\infty}([0, T] \times \reals^N \times \reals^D ; \reals)$ and are continuous in $(x,y)$ on compacts of $[0,T] \times \reals^N \times \reals^D$, uniformly in $t$, $i = 1, \dots, N, j = 1, \dots, D$. 
\item $z^\top a z \geq C |z|^2 $ for some constant $C>0$, for every $z \in \reals^{N}$ uniformly in $(t,x,y)$.
\end{enumerate}
\end{assumptionm}

%-------Assumption: Ratio density
\begin{assumptionm}
\label{ass:multi_ratio_density}
Recall $p(r, y)$ is the density of $V_r$.
\begin{align*}
	\left |  \frac{\sum_{l = 1}^D \partial_{y_l}(p(r, y) \beta_{l,k}(r, y))}{p(r, y)}  \right | \leq C_k \left ( \frac{|y|^{p_1}}{r^{q_1}} + \frac{|y|^{p_2}}{r^{q_2}} \right ),
\end{align*}
where $p_i \geq 0, q_i \in \reals$ and $p_i = 0$ implies $q_i \leq 0$, for $i = 1, 2$.
\end{assumptionm}

In the univariate case, our main innovation in the proofs from \Cref{sec:proofs} came from handling the technicalities associated with conditioning on the $\sigma$-algebra $\bar \FF_{t, T}^{V, B}$ and subsequently utilising the Brownian motion $\mrB$ as the stochastic integrator. This technicality led us to enforce \Cref{ass:ratio_density} on the density of $V_r$ to ensure our results hold in the univariate case. It should not come as a surprise that \Cref{ass:multi_ratio_density} is the correct counterpart in the multivariable scenario.

The extension of our main results from \Cref{sec:mainresults} to the higher dimensional case is straightforward. Indeed, one simply follows the methods of the proofs in \Cref{sec:proofs} and changes the univariate objects to their multivariable ones. Hence, we state the following results without proof.

%-------Theorem: The well-posed multi SPDE solution exists
\begin{theorem}
\label{thm:spdeexistencemulti} 
There exists a unique solution $u(t,x)$ to the SPDE \cref{eqn:spdewellposedmulti}, adapted to $(\bar \FF_{t,T}^{V, B})_{\tinT}$. Moreover, $t \mapsto u(t,x)$ belongs to $L^2(\vep, T ; H^1(\reals^N)) \cap C([\vep, T]; L^2(\reals^N))$ for all $\vep > 0$, $\Qro$ a.s.
\end{theorem}

%-------Theorem: multi Conditional Feynman-Kac formula
\begin{theorem}
\label{thm:conditionalfeynmanwellposedmulti} 
Let $u(t,x)$ be the unique $(\bar \FF_{t,T}^{V, B})_{\tinT}$-adapted solution to the SPDE \cref{eqn:spdewellposedmulti}. Then for all $t \in (0, T]$, $u(t,x)$ admits the representation
\begin{align*}
	u(t, x) = \Ex \big [ \varphi(X_T) | X_t = x, \bar \FF_{t,T}^{V,B}]
\end{align*}
$\dd x \times \dd \Qro$ a.e.
\end{theorem}

\begin{remark}
\label{remark:spdeinformalmulti}
As in the two-dimensional setting, an informal SPDE counterpart to the multivariable well-posed SPDE \cref{eqn:spdewellposedmulti} can be stated, namely 
\begin{align}
\begin{split}
	-\dd u (t, x) &= \left (\LL^x_t - \CC^x_t \right ) u(t,x)\dd t + \sum_{k = 1}^D \left (\BB^x_t \right )_k u(t,x) \bd B^k_t, \label{eqn:spdeinformalmulti} \\
	 u(T,x) &= \varphi(x).
\end{split}
\end{align}
\end{remark}
%------------------------------------------------------------------------Numerical analysis---------------------------------------------------------------------------------
%------------------------------------------------------------------------------------------------------------------------------------------------------------------------------
\section{Numerical analysis}
\label{sec:numerics}
\noindent
In this section, we develop a mixed Monte-Carlo PDE numerical method for the pricing of European put options by utilising our conditional Feynman-Kac formula (\Cref{thm:conditionalfeynmanwellposed}). Through our mixed Monte-Carlo PDE method, we will be able to achieve dimension and variance reduction as compared to a Full Monte-Carlo simulation or deterministic PDE numerical method by offloading the spot simulation onto a numerical PDE solver, and then handling the volatility process through Monte-Carlo simulation. Rather than utilising the well-posed SPDE \cref{eqn:spdewellposed} whose solution can be expressed as a suitable conditional expectation via our conditional Feynman-Kac formula, we will instead utilise the informal SPDE \cref{eqn:spdeinformal}. Briefly speaking, this is possible since time will be discretised, and thus there is no danger of any ill-posed stochastic integral arising. To further elaborate, first suppose we do decide to use the well-posed SPDE to develop our mixed Monte-Carlo PDE numerical method, and consider the following. We note that the coefficients in the well-posed SPDE \cref{eqn:spdewellposed} depend on $V_t$, thus we must first simulate $V$ from \cref{eqn:system2V}, and this itself requires simulation of the Brownian motion $B$. Then to numerically solve the well-posed SPDE \cref{eqn:spdewellposed} through finite difference we are required to simulate the backward Brownian motion $\mrB$. The crucial point is that $B$ and $\mrB$ are not the same, and in fact are related by \cref{eqn:mrB}. Lastly, by plugging in the increments of $\mrB$ into the well-posed SPDE \cref{eqn:spdewellposed} (after time discretisation), we then end up with the time discretised version of the informal SPDE \cref{eqn:spdeinformal}. Hence, it is simpler, more intuitive and equivalent to consider the informal SPDE for numerical purposes. For this reason, in this section, we only refer to the informal SPDE, and here on in will simply refer to it as the SPDE.

For convenience, we can formulate an informal version of the conditional Feynman-Kac formula in two dimensions (\Cref{thm:conditionalfeynmanwellposed}). Let $\bar u(t, x) = \Ex \big [\varphi(X_T) |X_t = x, \bar \FF_{t,T}^{V,B} \big ]$ where we refer to objects defined from \Cref{sec:preliminaries}. Then $\bar u(t,x)$ solves the informal SPDE
\begin{align}
\begin{split}
	-\dd u(t, x) &= \left (\LL^x_t -  \CC^x_t \right ) u(t,x)\dd t + \BB^x_t u(t,x) \bd B_t, \\
	u(T,x) &= \varphi(x), \label{eqn:spdenumerics}
\end{split}
\end{align}
where the (stochastic) differential operators $\LL_t^x, \BB_t^x, \CC_t^x$ are given in \crefrange{eqn:stochdiffoperatorL}{eqn:stochdiffoperatorC}. Denoting by $H$ the price of a European derivative which pays $\varphi(X_T)$ at time $T$, then $H_t = e^{-\int_t^T \intr_r \dd r} \Ex \big [\bar u(t, X_t)| X_t, V_t \big ]$, where $(\intr_t)_{\tinT}$ is the deterministic interest rate. Moreover, by following the strategy outlined in \Cref{remark:time0issue}, we are able to legitimately develop a mixed Monte-Carlo PDE method for pricing at time $t = 0$. Lastly, we remark that the methodology developed and examples considered in this section can be generalised to the higher dimensional framework by appealing to \Cref{remark:spdeinformalmulti}.
%------------------------------------------------------------------------Numerical schemes--------------------------------------------------------------------------------
\subsection{Numerical SPDE schemes}
Consider a time grid $\{0 = t_0 < t_1 < \cdots < t_n = T\}$ and space grid $\{x_{\text{min}} <  \cdots < x_{\text{max}}\}$, with $\Delta t := t_{i+1} - t_i$ and $\Delta x := x_{j+1} - x_j$. Let $u^{i,j} \equiv u(t_i, x_j)$. Define the following:
\begin{align*}
	\LL_i^j [u]  &:= \frac{1}{2} (\sigma^{i,j})^2\left ( \frac{u^{i,j+1} - 2u^{i,j} + u^{i, j-1}}{(\Delta x)^2}\right ) + \mu^{i,j} \left (\frac{u^{i, j+1} - u^{i,j}}{\Delta x } \right ), \\
	\BB_i^j[u] &:=  \rho_i \sigma^{i,j}\left (\frac{u^{i, j+1} - u^{i,j}}{\Delta x } \right ),\\
	\CC^j_i[u] &:=  \rho_i \beta^i \sigma_y^{i,j}  \left ( \frac{u^{i, j+1} - u^{i,j}}{\Delta x}\right ).
\end{align*}
Here it is clear that for example, $f^{i,j} \equiv f(t_i, x_j, V_{t_i})$. The SPDE \cref{eqn:spdenumerics} yields the following numerical schemes: 
%---------------------Numerical schemes
\begin{itemize}
%-------Semi-implicit
\item Semi-implicit: 
\begin{align}
	u^{i,j} = u^{i+1,j}  + (\LL_{i}^j - \CC_{i}^j)[u] \Delta t +\BB_{i+1}^j [u] \Delta B_i, \quad u^{n,j} = \varphi(x_j). \label{eqn:schemesemi}
\end{align}
%-------Crank-Nicolson
\item Crank-Nicolson:
\begin{align}
	u^{i,j} = u^{i+1,j}  + \frac{1}{2} \big ((\LL_i^j + \LL_{i+1}^j)[u] - (\CC_i^j + \CC_{i+1}^j)[u]\big ) \Delta t +\BB_{i+1}^j[u] \Delta B_i, \quad u^{n,j} = \varphi(x_j). \label{eqn:schemecn}
\end{align}
\end{itemize}
Note that one must take the right end point when discretising the backward stochastic integral.

%-------Lemma: Mixed Monte-Carlo / PDE method
\begin{lemma}[Mixed Monte-Carlo PDE method]
\label{lemma:mixedMCPDE}
Let $x$ be the initial point of $X$ and suppose it corresponds to the space point $x_{\hat m}$ for some $\hat m \in \integers$. A mixed Monte-Carlo PDE method to simulate $H_0$ is the following:
\begin{enumerate}
\item Simulate a path of $B$ and $V$  to obtain the observations $B_1 \dots, B_n$ and $V_1, \dots, V_n$.
\item For these given paths, numerically solve the SPDE to obtain the value $u^{0,\hat m}$, which is an observation of $u(0, x)$. 
\item Repeat steps (1) and (2) $M$ times to obtain observations $(u^{0,\hat m,k})_{1\leq k \leq M}$, where $u^{0,\hat m,k}$ denotes the $k$-th observation.
\item $H_0 = e^{-\int_0^T \intr_r \dd r}\,\Ex \left [\bar u(0, x) \right ] \approx e^{-\int_0^T \intr_r \dd r} \frac{1}{M} \sum_{k=1}^M u^{0,\hat m,k}$.
\end{enumerate}
\end{lemma}

%-------------------------------------------------------------------Numerical implementation-----------------------------------------------------------------------------
\subsection{Numerical implementation}
We consider pricing a European put option within the Inverse-Gamma model with constant parameters, see \citep{langrene2016switching}:
\begin{align}
	\dd S_t &= \intr S_t \dd t + S_t V_t \dd W_t , \quad S_0, \label{eqn:systemspotX} \\
	\dd V_t &= \kappa(\theta - V_t) \dd t +  \lambda  V_t \dd B_t, \quad V_0 = v_0, \label{eqn:systemspotV} \\
	\dd \langle W, B \rangle_t &= \rho \dd t. \nonumber
\end{align}
For simplicity we assume that the parameters $\kappa, \theta$ and $\lambda$ are strictly positive, so that the process $V$ is strictly positive, see \citep[][eq. 0.2]{zhao2009inhomogeneous}.
Let $X_t = \ln(S_t/K)$, where $K$ is the strike of a European put option on $S$. We can rewrite the system \crefrange{eqn:systemspotX}{eqn:systemspotV} as
\begin{align}
	\dd X_t &= \left (\intr -\frac{1}{2} V_t^2 \right) \dd t + V_t \dd W_t , \quad X_0 = \ln(S_0/K),  \label{eqn:systemlogspotX} \\
	\dd V_t &= \kappa (\theta - V_t) \dd t +  \lambda  V_t \dd B_t, \quad V_0 = v_0,  \label{eqn:systemlogspotV} \\
	\dd \langle W, B \rangle_t &= \rho \dd t. \nonumber
\end{align}
For numerical purposes, we will instead consider the system \crefrange{eqn:systemlogspotX}{eqn:systemlogspotV}.

Let $\varphi^P(x) = K(1-e^x)_+$ and $u^P(t,x) = \Ex \big[\varphi^P(X_T) |X_t = x, \bar \FF_{t,T}^{V,B} \big ]$. Then $u^P$ solves the SPDE \cref{eqn:spdenumerics} with terminal condition $\varphi^P$, where 
\begin{align*}
	\mu(t,x,V_t) &= \intr - \frac{1}{2} V_t^2, 		&  \sigma(t,x,V_t) &= V_t, 	& \alpha(t,V_t) &= \kappa (\theta - V_t),
	 & \beta(t, V_t) &= \lambda V_t.
\end{align*}
Thus, the time $t$ price of a put option on $S$ is given by $H^P_t := e^{-\intr (T-t)} \Ex [ u^P(t,X_t) | X_t, V_t]$. Moreover, it is straightforward to see that the right and left boundary conditions of the SPDE for $u^P$ are
\begin{align*}
	\lim_{x \to \infty} u^P(t,x) &= 0, \\
	\lim_{x \to - \infty} u^P(t,x) &= K,
\end{align*}
respectively.

%--------Remark: system does not satisfy Assumptions A - D.
\begin{remark} 
We briefly comment on the how the system \crefrange{eqn:systemlogspotX}{eqn:systemlogspotV} and put option payoff $\varphi^P$ handles \Crefrange{ass:sde}{ass:ratio_density}. First note that the system \crefrange{eqn:systemlogspotX}{eqn:systemlogspotV} possesses a pathwise unique strong solution, as \cref{eqn:systemlogspotV} satisfies \Cref{ass:sde} and \cref{eqn:systemlogspotX} is really just a formula for $X$ in terms of $V$. Moreover, $V_0$ is degenerate and $\beta(t, y) = \lambda y$, and thus \Cref{ass:mrB} is satisfied. More importantly, the system \crefrange{eqn:systemlogspotX}{eqn:systemlogspotV} does not seem to satisfy all the criteria in \Cref{ass:spde}. However, \Cref{ass:spde} is really stronger than what is required, and relaxations can be made provided that one includes various approximation and truncation procedures in the relevant proofs, not dissimilar to the case of deterministic PDEs. However, in \Cref{sec:proofs} we have evidently chosen not to prove our results in such generality, so as to keep the (already quite technical) proofs as simple as possible, and to ensure that the main ideas are not lost. 
%see \citep[][Proof of Theorem 2.1]{pardoux1979stochastic} 
For example, Assumptions \cref{ass:C1} and \cref{ass:C2} can clearly be circumvented through standard localisation arguments. Assumption \cref{ass:C3} is in fact satisfied by the system \crefrange{eqn:systemlogspotX}{eqn:systemlogspotV}. Lastly, due to the linear structure of the SDE \cref{eqn:systemlogspotV}, an explicit form for the pathwise unique strong solution of it exists \citep[][eq. 0.2]{zhao2009inhomogeneous}, and from this it is straightforward to deduce that the solution remains strictly positive. However, it is not lower bounded by a strictly positive constant. Despite this, the uniform ellipticity condition \cref{ass:C4} can be circumvented by replacing the SDE for $V$ in \cref{eqn:systemlogspotV} with 
\begin{align*}
	\dd \bar V_t = \kappa (\theta - (\bar V_t - \vep))\dd t + \lambda (\bar V_t - \vep)\dd B_t, \quad \bar V_0 = v_0,
\end{align*}
for some $\vep \leq v_0$, and thus one obtains the lower bound $\bar V_t \geq \vep$. By doing so we satisfy the uniform ellipticity condition \cref{ass:C4} as $\sigma^2(t, x, \bar V_t) \geq \vep^2$. Moreover, adding in this artificial lower bound will not change numerical experiments when $\vep$ is close to $v_0$. Finally, we are unfortunately unable to verify if \cref{eqn:systemlogspotV} satisfies \Cref{ass:ratio_density}, as this would require stringent quantitative results on the density of $V_r$ and its derivative. It is actually possible to find an explicit expression for the density of $V_r$, see \citep[][Theorem 2.8]{zhao2009inhomogeneous}, however this representation is rather complicated and difficult to work with. Despite this, we conjecture that \Cref{ass:ratio_density} holds for our example, and the validity of the numerical implementation is evidenced by our results comparing the mixed Monte-Carlo PDE method with the other two Monte-Carlo methods below.
\end{remark}

We will compare our mixed Monte-Carlo PDE method with the usual Full (two-dimensional) Monte-Carlo method by computing implied volatility for a 6M ATM European put option, and then investigating the accuracy and speed by varying the number of paths and time steps for both methods. As the benchmark for comparison, we will utilise the so-called Mixing Solution relationship, see \citep{das2022closed}. This relationship states that European put/call option prices can be expressed as an expectation of a functional of the volatility/variance process, this functional being essentially a Black-Scholes formula. We will state the result without proof, as it is a clear adaptation of the derivation for the Black-Scholes formula.
\begin{lemma}[Mixing Solution]
Let $\NN(x) = \int_{-\infty}^x  \frac{1}{\sqrt{2 \pi}} e^{-y^2/2} \dd y$ denote the standard normal distribution function. Then
\begin{align*}
	H_0^P &= \Ex \left [\Ex \left [ e^{-\intr T } (K  - S_T)_+ | \FF_T^B \right ]\right ] \\
	&=\Ex \left [\text{Put}_{\text{BS}}\left (S_0 \xi_T, (1 - \rho^2)  \int_0^T V^2_r \dd r \right) \right ],
\end{align*}
where 
\begin{align*}
		\xi_T  &= \exp \left ( \rho \int_0^T V_r \dd B_r - \frac{\rho^2}{2} \int_0^T  V^2_r \dd r\right ),
\end{align*}
and
\begin{align*}
	\text{Put}_{\text{BS}}(x,y) &:= K e^{-\intr T} \NN(- d_-) - x \NN(- d_+), \\
	d_\pm(x,y) := d_{\pm} &:= \frac{\ln(x/K) + \intr T}{\sqrt{y}} \pm \frac{1}{2} \sqrt{y}.
\end{align*}
\end{lemma}
The advantage of utilising the Mixing Solution relationship numerically is that it requires only a one-dimensional Monte-Carlo simulation, and hence is superior in terms of efficiency than the Full Monte-Carlo method. Moreover, it converges faster, which is a simple consequence of the law of total variance. Of course, the Mixing Solution relationship only works for European options, and only for models where the spot satisfies an SDE of the form \cref{eqn:systemspotX}. The method of numerically pricing options via the Mixing Solution will be called the Monte-Carlo Mixing Solution method. 

The (constant) parameters utilised in all our numerical experiments are given in the following table:
\begin{center}
\begin{tabular}{lllll l llll} 
\toprule
$S_0$ & \ $V_0$ & $T$ & $K$ & \ $\intr$ && $\kappa$ & $\theta$ & $\lambda$ & \; \ $\rho$ \\
\midrule
$100$ & \ $20\%$ &  6M  & ATM & \ 1\% && $5.00$ & $18\%$ & $0.90$ & $-0.35$  \\
 \bottomrule \\
\end{tabular}
\end{center}

For the mixed Monte-Carlo PDE method, to numerically solve the SPDE we utilise the Crank-Nicolson scheme \cref{eqn:schemecn} with the following space parameters, which will remain fixed throughout all our experiments:

\begin{center}
\begin{tabular}{cccc} 
\toprule
$x_0$ & \ $x_{\text{min}}$ & \ $x_{\text{max}}$ & \#Space points   \\
\midrule
$\ln(S_0/K)$  & \ $x_0 - 4V_0 \sqrt{T}$ & $x_0 + 4V_0 \sqrt{T}$ & $250$ \\
 \bottomrule \\
\end{tabular}
\end{center}

The benchmark will be given via the Monte-Carlo Mixing Solution method, where we utilise 1,000,000 paths, with 24 time steps per day, where a year is comprised of 253 trading days. 

%--------Remark: Routines found on GitHub
\begin{remark}
The python code utilised for all our numerical experiments can be found on GitHub \citep{das2023}. In particular, what is provided are:
\begin{itemize}
\item Routines which compute European put/call option prices via the Monte-Carlo Mixing Solution method, Full Monte-Carlo method and our mixed Monte-Carlo PDE method.
\item A routine which compares the runtimes and errors in the aforementioned methods.
\end{itemize}
\end{remark}

%------ImpVol Figure
\begin{center}
\begin{figure}[H]
\includegraphics[width=12cm]{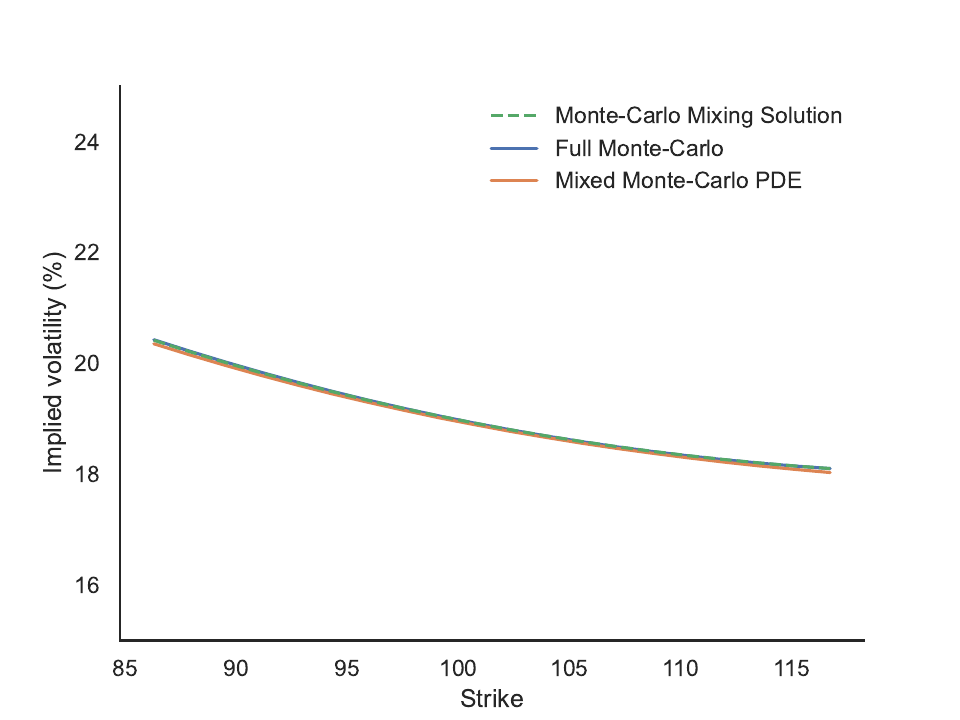}
\caption{The implied volatility curve in the Inverse-Gamma model. The number of Monte-Carlo paths for the Monte-Carlo Mixing Solution, Full Monte-Carlo and mixed Monte-Carlo PDE methods are $10 \times 10^5, 15 \times 10^5, 10 \times 10^4$ respectively, whereas the number of time steps are $24, 48$ and $1$ per day respectively.}
\label{fig:impvol}
\end{figure}
\end{center}
\Cref{fig:impvol} shows a plot of the implied volatility curve obtained from all three methods in the Inverse-Gamma model with the aforementioned parameters. One can see qualitatively that the mixed Monte-Carlo PDE method does indeed reproduce the implied volatility curve well. More detailed and quantitative numerical results are provided in \Cref{appen:numericalresults}.

%-------------------------------------------------------------------Numerical results------------------------------------------------------------------------------
One will note that for the two methods, there is ostensibly a mismatch between the number of time-steps per day and paths chosen in our numerical experiments in \Cref{table:IGaimpvolsSP,table:IGaimpvolsFullMC}. However, this is not necessarily the case. First, it does not seem appropriate to directly compare the number of time-steps utilised by these two methods, since the mixed Monte-Carlo PDE method requires a time discretisation of $V$ as well as the SPDE, however the Full Monte-Carlo method requires a time discretisation of both $V$ and $X$. Secondly, the apparent mismatch between the number of paths considered for the two methods can be easily clarified as well. Via properties of conditional expectation, one can show that given a number of paths, the Monte-Carlo standard error for the mixed Monte-Carlo PDE method is significantly less than that of the Full Monte-Carlo method. Intuitively this makes sense; simulation of $X$ usually contributes the most to the Monte-Carlo variance, however in our mixed Monte-Carlo PDE method we bypass simulation of $X$ by offloading it to the PDE component. In fact this highlights a substantial advantage of our mixed Monte-Carlo PDE method; bluntly speaking the PDE component does the hard work by handling $X$, whereas the Monte-Carlo component does the easier work by tackling $V$.

At first glance it may seem that the run times of the mixed Monte-Carlo PDE method pale in comparison to the Full Monte-Carlo method. However these are not at all comparable, as another significant advantage of the mixed Monte-Carlo PDE method is that as it is a PDE method, we obtain the price of the put option for various $S_0$ values (250 values in this case!), whereas the Full Monte-Carlo method only obtains it for a single value.

For the mixed Monte-Carlo PDE method, we have considered a special case where we utilise 1,000,000 paths for each choice of \#Steps/day. This is in an attempt to reduce the Monte-Carlo standard error sufficiently low so that it is negligible compared to the time and space discretisation error, thereby giving us a better idea of what the combined time and space discretisation errors solely are. For the Full Monte-Carlo method, we have proceeded in a similar manner, where we have considered a case with 10,000,000 paths for each choice of \#Steps/day.

As mentioned above, it is difficult to compare the errors between the two methods as their number of time-steps per day and paths do not have a direct correspondance. However, we have selected them as best as we believe possible in order to draw a fair comparison. The Full Monte-Carlo errors in \Cref{table:IGaimpvolsFullMC} are standard and require no further investigation. For the mixed Monte-Carlo PDE method results in \Cref{table:IGaimpvolsSP}, the absolute errors and standard errors are at most approximately 10 basis points, which is more than sufficient in application. One thing to note is that it seems to have an unpredictable error for \#Steps/day = 0.5, meaning that the absolute error is not decreasing very monotonically as the number of paths increase. However, it starts to settle down for \#Steps/day = 1, 2. It seems logical to attribute this consistency to the PDE solver being sufficiently accurate on these finer time grids.

%-------------------------------------------------------------------------------Conclusion--------------------------------------------------------------------------------
%------------------------------------------------------------------------------------------------------------------------------------------------------------------------------
\section{Conclusion}
\label{sec:conclusion}
\noindent 
In this article we have proved a conditional Feynman-Kac formula which arises in the context of mathematical finance, and proved under certain assumptions that the existence and uniqueness of the associated SPDE is valid. These results are similar to results obtained in Section 6 of \citep{pardoux1982equations}, however in our case, non-trivialities arise due to the backward Brownian motion and backward filtration that must be considered, namely $\mrB$ and $(\bar \FF_{t,T}^{V,B})_{\tinT}$. Under additional assumptions on the speed of growth of the density of the auxiliary process $V$, we have shown that Pardoux's results can be adapted to the setting considered in this article. The purpose of developing this conditional Feynman-Kac formula is to utilise it to solve problems in mathematical finance. Indeed, we demonstrate its application in the simple setting of pricing a European put option in the Inverse-Gamma model. The conditional Feynman-Kac formula can be applied in other settings in mathematical finance, for example, mixing Least Square Monte-Carlo methods with numerical PDE methods, which will be the focus of forthcoming articles. 

%\subsection*{Declaration of interest}
%Declarations of interest: none.
%
\subsection*{Funding}
K. Das and I. Guo have been supported by the Australian Research Council (Grant DP220103106). I. Guo was also partially supported by CSIRO Data61 Risklab. During this project, the Centre for Quantitative Finance and Investment Strategies has been supported by BNP Paribas.

\subsection*{Acknowledgements}
The authors would like to thank two anonymous referees for their valuable comments and insights.

%--------------------------------------------------------------------------------References--------------------------------------------------------------------------------
\renewcommand{\bibname}{References}
\bibliographystyle{plainnat}
\bibliography{refs}

%------------------------------------------------------------------------------MSC/ACM classes----------------------------------------------------------------------------
%----------------------MSC
% 60H15
% 60H30
% 91G20
% 91G60

%----------------------ACM
% G.3

%--------------------------------------------------------------------------------Appendix----------------------------------------------------------------------------------
%-----------------------------------------------------------------------------------------------------------------------------------------------------------------------------
%-----------------------------------------------------------------------------------------------------------------------------------------------------------------------------
\appendix
\crefalias{section}{appendix}
\crefname{appendix}{Appendix}{Appendices}
\Crefname{appendix}{Appendix}{Appendices}
%\clearpage
%-------------------------------------------------------------------Backwards stochastic integration-----------------------------------------------------------------------
\section{Some content on backward stochastic calculus}
\label{appen:backwarddefns}
\noindent
In this appendix, we provide the definitions of the backward versions of common objects and concepts from stochastic calculus. These definitions are straightforward counterparts to their forward versions. For this reason, this content has sometimes been dubbed backward stochastic calculus. However, we should stress that backward stochastic calculus should not be confused with the theory of backward stochastic differential equations developed by Pardoux and Peng, the latter being quite prevalent in the current stochastic analysis literature.
%-------Defn: Backwards filtration
\begin{definition}[Backward filtration]
Let $(\GG_{t,T})_{\tinT}$ be a decreasing collection of $\sigma$-algebras. Then $(\GG_{t,T})_{\tinT}$ is called a backward filtration. We assume all backward filtrations considered satisfy the usual conditions, which for backward filtrations are: left continuity, i.e., $\GG_{t, T} = \bigcap_{\vep > 0} \GG_{t - \vep, T}$ for all $\tinT$, and also that $\GG_{T,T}$ is augmented by null sets.
\end{definition}

%-------Defn: Backward martingale
\begin{definition}[Backward martingale]
\label{defn:backwardsmg}
%Let $\FF_{t,T}= \sigma(X_v - X_u, t \leq u < v \leq T)$. Then $(\FF_{t,T})_{\tinT}$ is a backward filtration. 
Consider a process $M$ as well as a backward filtration $(\GG_{t,T})_{\tinT}$. Suppose $M$ satisfies the following.
\begin{enumerate}[label = (\roman*)]
\item $M$ is adapted to the backward filtration $(\GG_{t,T})_{\tinT}$.
\item $\Ex|M_t| < \infty $ for all $t \in [0,T]$. 
\item $\Ex[M_s | \GG_{t,T}] = M_t$ for $s < t$.
\end{enumerate}
Then $M$ is called a backward martingale w.r.t. the backward filtration $(\GG_{t,T})_{\tinT}$. 
\end{definition}

%-------Defn: Backward stopping time
\begin{definition}[Backwards stopping time]
\label{defn:backwardlmg}
Consider a backward filtration $(\GG_{t,T})_{\tinT}$. The random variable $\tau: \Omega \to \reals$ is called a backward stopping time if the events $\{\tau \geq t \} \in \GG_{t,T}$ for each $t$. 
\end{definition}

%-------Defn: Backward local martingale
\begin{definition}[Backward local-martingale]
\label{defn:backwardlmg}
Consider a process $M$ which is adapted to a backward filtration $(\GG_{t,T})_{\tinT}$. Let $(\tau_n)_n$ be a sequence of backward stopping times with respect to $(\GG_{t,T})_{\tinT}$ such that 
\begin{enumerate}[label = (\roman*), ref = \roman*]
\item $\tau_n \downarrow 0$ a.s.
\item $(\tau_n)_n$ is non-increasing a.s.
\end{enumerate}
Suppose that $M_t^{(n)} := M_{t \vee {\tau_n}}$ is a $(\GG_{t,T})_{\tinT}$ backward martingale for each $n$. Then $M$ is called a backward local-martingale relative to $(\GG_{t,T})_{\tinT}$.
\end{definition}

%-------Defn: Backward Brownian motion
\begin{definition}[Backward Brownian motion]
\label{defn:backwardbm}
Consider a process $Z$ taking values in $\reals^d$ which is adapted to a backward filtration $(\GG_{t,T})_{\tinT}$. In addition, let $Z$ satisfy the following:
\begin{enumerate}[label = (\roman*)]
\item $Z$ is continuous in $t$ a.s.
\item For $t > s$, the increment $Z_s - Z_t \sim \NN(0, (t-s)I)$ where $I$ is the $d \times d$ identity matrix. 
\item For $t >  s$, the increment $Z_s -Z_t$ is independent of $\GG_{t,T}$. 
\end{enumerate}
Then $Z$ is called a backward Brownian motion relative to $(\GG_{t,T})_{\tinT}$. Moreover, if $Z_T = 0$, then $Z$ is called a standard backward Brownian motion relative to $(\GG_{t,T})_{\tinT}$.
\end{definition}

%-------Remark: B BM is B MG
\begin{remark}
It is clear that a backward Brownian motion is a backward martingale.
\end{remark}

%-------Remark: B BM is B MG
\begin{remark}
It is clear that Levy's characterisation of Brownian motion extends to the backward scenario. Namely, a stochastic process is a backward Brownian motion if and only if it is a backward local-martingale with quadratic variation $t$.
\end{remark}
The following theorem is crucial in this article. It states how to construct an appropriate backward Brownian motion when the backward filtration of interest has undergone a certain type of filtration enlargement.
\begin{theorem}[{\citep[][Theorem 2.2]{pardoux1986grossissement}}]
\label{thm:definingmrB}
Enforce \Cref{ass:multi_mrB}. Recall from \Cref{sec:multivariablesetting} that $\bar \FF^{V, B}_{t,T} := \FF^B_{t,T} \vee \sigma(V_t)$ and
\begin{align*}
	\mrB_t^k = B_t^k - B_T^k - \int_t^T \frac{\sum_{l = 1}^D \partial_{y_l}(p(r, V_r) \beta_{l,k}(r, V_r))}{p(r, V_r)} \dd r, \quad k = 1, \dots , D,  
\end{align*}
where the integrand is taken to be zero if ever $p$ is zero. Then $\mrB$ is a $\reals^D$ valued backward Brownian motion in $(\bar \FF_{t,T}^{V,B})_{\tinT}$.
\end{theorem}

%-------------------------------------------------------------------Numerical results------------------------------------------------------------------------------
\section{Numerical results}
\label{appen:numericalresults}
%--------------------------Benchmark table (mixing solution)
\begin{table}[H]
\caption{Implied volatility, Monte-Carlo standard error, and Run time for pricing an ATM Put option with maturity 6 months. Price is obtained via the Monte-Carlo Mixing Solution method with 1,000,000 paths and 24 time steps per day (Benchmark).}
\label{table:IGaimpvolsbenchmark}
\begin{center}
\begin{tabular}{cc c cccc}
\toprule
\multicolumn{7}{c}{Benchmark} \\
\cmidrule{1-2}\cmidrule{4-7} 
 \#Steps/day 	 &\#Path		& & IV(\%)	 & S.E.(bp)  	& Abs Err(bp)	& Run(s) \\
\midrule
     24& $ 10 \times 10^5$ 	&& 18.872   	& 1.20   		& N/A	& 226.7  	 \\
\bottomrule
\end{tabular}
\end{center}
\end{table}
%---------------------------------------------------------MC/PDE numerics
\begin{table}[H]
\caption{Implied volatilities, Monte-Carlo standard errors, Absolute errors, and Run times for pricing an ATM Put option with maturity 6 months via the mixed Monte-Carlo PDE method, where \# of paths and time steps per day are varied, and \# of space points is fixed at 250.}
\label{table:IGaimpvolsSP}
\begin{center}
\begin{tabular}{cc c cccc}
\toprule
\multicolumn{7}{c}{Mixed Monte-Carlo PDE} \\
\cmidrule{1-2}\cmidrule{4-7} 
 \#Steps/day 	 &\#Path		& & IV(\%)	 & S.E.(bp)  	& Abs Err(bp)	& Run(s) \\
\midrule
0.5   & $10 \times 10^3$   	&&	 	18.77   & 11.71  & 9.72 & 74.6 	 \\
   & $20 \times 10^3$   	&&	 	18.99   & 8.51  & 11.35 & 148.9 	 \\
   & $40 \times 10^3$   	&&	 	18.87   & 5.95  & 0.09 & 298.7 	 \\
      & $80 \times 10^3$   	&&	 	18.85   & 4.21  & 2.03 & 595.7 	 \\
        & $10 \times 10^5$   	&&	 	18.91   & 1.20  & 3.85 & 7404.3 	 \\
      \midrule
1   & $10 \times 10^3$   	&&	 	18.79   & 11.71  & 8.27 & 147.9 	 \\
   & $20 \times 10^3$   	&&	 	18.85   & 8.57  & 1.68 & 295.2 	 \\
    & $40 \times 10^3$   	&&	 	18.83   & 5.95  & 3.80 & 589.0 	 \\
   & $80 \times 10^3$   	&&	 	18.87   &  4.20 & 0.40 &  1177.0	 \\
      & $10 \times 10^5$   	&&	 	18.88   &  1.19 & 0.48 &  14712.6	 \\
\midrule
2   & $10 \times 10^3$   	&&	 	18.96   & 12.09  & 8.85 & 297.7 	 \\
   & $20 \times 10^3$   	&&	 	18.84   & 8.36   & 3.28 & 597.8 	 \\
    & $40 \times 10^3$   	&&	 	18.80   & 5.88  & 6.82 & 1184.0 	 \\
   & $80 \times 10^3$   	&&	 	18.87   &  4.23 & 0.02 &  2376.3	 \\
      & $10 \times 10^5$   	&&	 	18.89   &  1.19 & 1.52 &  29642.8	 \\
\bottomrule
\end{tabular}
\end{center}
\end{table}

%---------------------------------------------------------Full MC numerics
\begin{table}[H]
\caption{Implied volatilities, Monte-Carlo standard errors, Absolute errors, and Run times for pricing an ATM Put option with maturity 6 months via the Full Monte-Carlo method, where the number of paths and time steps per day are varied.}
\label{table:IGaimpvolsFullMC}
\begin{center}
\begin{tabular}{cr c cccc}
\toprule
\multicolumn{7}{c}{Full Monte-Carlo} \\
\cmidrule{1-2}\cmidrule{4-7} 
 \#Steps/day 	 &\#Path		& & IV(\%)	 & S.E.(bp)  	& Abs Err(bp)	& Run(s) \\
\midrule
0.5   & $40 \times 10^3$   	&&	 	18.89   & 14.55  & 2.20 & 0.20 	 \\
   & $80 \times 10^3$   	&&	 	19.09   & 10.34  & 22.08 & 0.41 	 \\
   & $160 \times 10^3$   	&&	 	18.93   & 7.27 & 5.66 & 1.24	 \\
      & $320 \times 10^3$   	&&	 	18.97   & 5.14  & 9.97 & 2.59 	 \\
            & $100 \times 10^5$   	&&	 	19.00   & 0.92  & 12.51 & 77.50	 \\
      \midrule
      
1   & $40 \times 10^3$   	&&	 	18.93   & 14.51  & 5.83 & 0.41 	 \\
   & $80 \times 10^3$   	&&	 	18.90   & 10.26  & 3.25 & 0.82 	 \\
   & $160 \times 10^3$   	&&	 	18.84   & 7.26 & 3.04 & 2.41	 \\
      & $320 \times 10^3$   	&&	 	18.93   & 5.13  & 6.32 & 4.83 	 \\
            & $100 \times 10^5$   	&&	 	18.95   & 0.92  & 7.48 & 156.52 	 \\
      \midrule
   
2   & $40 \times 10^3$   	&&	 	18.67   & 14.41  & 20.04 & 0.82 	 \\
   & $80 \times 10^3$   	&&	 	18.86   & 10.25  & 1.29 & 1.65 	 \\
   & $160 \times 10^3$   	&&	 	18.93   & 7.26 & 5.56 & 4.86	 \\
      & $320 \times 10^3$   	&&	 	18.93   & 5.14  & 6.09 & 9.61 	 \\
            & $100 \times 10^5$   	&&	 	18.91   & 0.92  & 3.57 & 310.92 	 \\
      \midrule
   
4   & $40 \times 10^3$   	&&	 	18.85   & 14.39  & 2.33 & 1.62 	 \\
   & $80 \times 10^3$   	&&	 	18.92   & 10.22  & 4.91 & 3.45 	 \\
   & $160 \times 10^3$   	&&	 	18.77   & 7.22 & 9.73 & 9.74	 \\
      & $320 \times 10^3$   	&&	 	18.89   & 5.12  & 2.30 & 19.18 	 \\
            & $100 \times 10^5$   	&&	 	18.89   & 0.92  & 1.38 & 624.10 	 \\
      \midrule
   
8   & $40 \times 10^3$   	&&	 	18.81   & 14.48  & 6.55 & 3.22 	 \\
   & $80 \times 10^3$   	&&	 	18.89   & 10.23  & 1.83 & 6.58 	 \\
   & $160 \times 10^3$   	&&	 	18.74   & 7.20 & 13.56 & 19.36	 \\
      & $320 \times 10^3$   	&&	 	18.82   & 5.11  & 4.89 & 38.27 	 \\
            & $100 \times 10^5$   	&&	 	18.88   & 0.92  & 0.84 & 1242.32 	 \\
      \midrule
      
      16   & $40 \times 10^3$   	&&	 	18.76   & 14.42  & 10.81 & 6.47 	 \\
   & $80 \times 10^3$   	&&	 	18.85   & 10.22  & 2.51 & 13.01 	 \\
   & $160 \times 10^3$   	&&	 	18.99   & 7.27 & 12.14 & 38.70	 \\
      & $320 \times 10^3$   	&&	 	18.93   & 5.13  & 5.91 & 76.65 	 \\
            & $100 \times 10^5$   	&&	 	18.85   & 0.92  & 1.91 & 2477.73 	 \\
      \midrule
   
24   & $40 \times 10^3$   	&&	 	18.86   & 14.40  & 1.18 & 9.63 	 \\
   & $80 \times 10^3$   	&&	 	18.88   & 10.25  & 0.40 & 19.58 	 \\
   & $160 \times 10^3$   	&&	 	18.98   & 7.28 & 10.56  & 57.93	 \\
      & $320 \times 10^3$   	&&	 	18.85   & 5.11  & 2.13 & 115.06 	 \\
            & $100 \times 10^5$   	&&	 	18.86   & 0.92  & 0.74 & 3718.76 	 \\
\bottomrule
\end{tabular}
\end{center}
\end{table}

\end{document}